\documentclass[11pt]{article}

\usepackage{amssymb}
\usepackage{amsmath}
\usepackage{amsfonts}
\usepackage{amsthm}
\usepackage{epsfig}
\usepackage{amscd}
\usepackage{amsfonts}
\usepackage{amsbsy}
\usepackage[shortlabels]{enumitem}
\usepackage{xcolor}
\usepackage[active]{srcltx}
\usepackage{latexsym}
\usepackage{graphicx}
\usepackage{psfrag}
\usepackage{color}
\usepackage{float}
\usepackage{comment}
\usepackage{fullpage}
\usepackage[normalem]{ulem} 
\usepackage{hyperref}
\usepackage{cite}
\usepackage{rotating}
\usepackage{caption}
\newtheorem {theorem} {Theorem}
\newtheorem {proposition} [theorem]{Proposition}
\newtheorem {corollary} [theorem]{Corollary}
\newtheorem {lemma}  [theorem]{Lemma}

\newtheorem {remark} [theorem]{Remark}

\newcommand{\R}{{\mathbb R}}

\newcommand{\V}{{\cal{V}}}
\newcommand{\U}{{\cal{U}}}

\begin{document}
\title{Kahan--Hirota--Kimura maps preserving\\ original cubic Hamiltonians}
\author{ V\'{\i}ctor Ma\~{n}osa$^{(1)}$ and Chara Pantazi$^{(2)}$
    \\*[0.1 truecm]
    \\*[-.1 truecm] {\small \textsl{$^{(1)}$ Departament de Matem\`{a}tiques (MAT),}}
     \\*[-0.1 truecm] {\small \textsl{Institut de Matem\`{a}tiques de la UPC-BarcelonaTech (IMTech),}}
    \\*[-0.1 truecm] {\small \textsl{Universitat Polit\`{e}cnica de Catalunya-BarcelonaTech (UPC)}}
    \\*[-0.1 truecm] {\small \textsl{Colom 11, 08222 Terrassa, Spain}}
    \\*[-0.1 truecm] {\small \textsl{victor.manosa@upc.edu}}
        \\*[-0.1 truecm] {\small \textsl{$^{(2)}$ Departament de Matem\`{a}tiques (MAT),}}
    \\*[-0.1 truecm] {\small \textsl{Universitat Polit\`{e}cnica de Catalunya-BarcelonaTech (UPC)}}
    \\*[-0.1 truecm] {\small \textsl{Doctor Mara\~{n}\'on 44-50,
08028 Barcelona, Spain}}
    \\*[-0.1 truecm] {\small \textsl{chara.pantazi@upc.edu}}}

\maketitle

\abstract{We study the class of cubic Hamiltonian vector fields whose associated Kahan--Hirota--Kimura (KHK) maps preserve the original Hamiltonian function. Our analysis focuses on these fields in $\mathbb{R}^2$ and $\mathbb{R}^4$, extending to a family of fields in $\mathbb{R}^6$. Additionally, we investigate various properties of these fields, including the existence of additional first integrals of a specific type, their role as Lie symmetries of the corresponding 
KHK map, and the symplecticity of these maps.}

\section{Introduction}\label{s:intro}

The Kahan--Hirota--Kimura discretization method (KHK, from now on) is a numerical method proposed, independently, by Kahan in \cite{kahan1993,KL97} and Hirota and Kimura \cite{hirota2000,kimura2000} to integrate quadratic vector fields. It is a one-step method, which is linearly implicit, and whose inverse is also linear implicit, so it defines a birational map $\Phi_h(\mathbf{x})$, where $\mathbf{x}\in\mathbb{R}^n$, which can be written as follows (see \cite{Celledoni2013}, for instance): given the vector field $X=\sum_{i=1}^n X_i\partial/\partial x_i$ with associated differential system
$$
\dot{\mathbf{x}}=X(\mathbf{x}),
$$
its associated KHK map is
\begin{equation}\label{e:phi}
\Phi_h(\mathbf{x})=\mathbf{x}+h\left(I-\frac{1}{2}h {\mathrm D}X({\bf{x}})\right)^{-1} X({\bf{x}}),
\end{equation}
where ${\mathrm D}X$ is the differential matrix of the field.

This method has been shown to be especially effective for the numerical integration of integrable quadratic systems in any finite dimension. This is because, on many occasions, the resulting associate discrete system (the KHK map), in turn, admits a first integral, see  \cite{Celledoni2013,Celledoni2019b} for instance. The
KHK maps are also important in the context of the theory of discrete integrability because, as other families like the QRT maps  introduced by Quispel, Roberts, and Thompson at the end of the 1980s \cite{QRT1,QRT2},  they exhibit an important display of geometric and algebraic-geometric properties related with their dynamics \cite{ASW,Celledoni2013,Celledoni2014,Cel17,
Celledoni2019b,
Kam19,Petrera2011,Petrera19-1,PS10,Petrera19-2,Petrera19-3,PZ17}.

In this paper, we will consider quadratic differential systems (vector fields) in $\R^n$  with $n=2,4,6$ possessing a cubic Hamiltonian function.  In \cite{Celledoni2013} it is proved that if the vector field has a cubic Hamiltonian, then its associated KHK map has a conserved quantity or first integral:

\begin{theorem}[Celledoni, McLachlan, Owren, Quispel, \cite{Celledoni2013}]\label{t:teoham}
For all cubic Hamiltonian systems of the form $X_H=K\nabla H$, where $K$ is a constant anti-symmetric matrix, on symplectic vector spaces and on all Poisson vector spaces with constant Poisson structure,
the KHK method has a first integral given by the modified Hamiltonian
\begin{equation}\label{e:modHam}
 \widetilde{H}({\bf{x}})=H({\bf{x}})+\frac{1}{3}h\nabla H({\bf{x}})^T\left(I-\frac{1}{2}h {\mathrm D}X({\bf{x}})\right)^{-1} X({\bf{x}}).
\end{equation}

\end{theorem}
Recall that a map $\Phi: \U\subset\R^n\rightarrow \R^n$ admits a \emph{first integral}  $H:\V\subseteq \U\longrightarrow \mathbb{R}$  defined in an open set $\V$ if $H(\Phi({\bf{x}}))=H({\bf{x}}),$ for all ${\bf{x}}\in \V.$

As a consequence of the above result, for a vector field with a Hamiltonian $H$ satisfying the hypothesis of Theorem \ref{t:teoham},  $H$ is also a first integral of its associated KHK map if and only if it holds

\begin{equation}\label{e:restriction}
  \nabla H({\bf{x}})^T\left(I-\frac{1}{2}h{\mathrm D}X({\bf{x}})\right)^{-1} X({\bf{x}})=0,
\end{equation}
in this case we say that \emph{the KHK map preserves the original Hamiltonian}.

Our objective is to investigate which vector fields with a cubic Hamiltonian have their original first integral preserved under the KHK discretization. Our analysis will be limited to even-dimensional spaces with the canonical Poisson structure, \cite{Celledoni2013}.  As we will see, the set of quadratic Hamiltonian vector fields satisfying \eqref{e:restriction} is non-empty and, for two or more degrees of freedom, non-trivial.
Our initial purpose was to examine the relation between the fact that the KHK maps preserve a given cubic Hamiltonian and the fact that the original Hamiltonian vector field is a Lie symmetry of the map (see definition in Section 2). Although all the fields that we have found with one or two degrees of freedom are, indeed, Lie symmetries, we have been able to detect several quadratic Hamiltonian fields with three degrees of freedom, whose KHK maps preserve the initial Hamiltonian but which are not Lie symmetries
of these maps.

The main results are Theorems \ref{t:Hamiltonians de R2} and \ref{t:Hamiltonians de R4}, where the set of cubic Hamiltonian vector fields  with one or two degrees of freedom and whose associated KHK maps preserve the original Hamiltonian are characterized, and Proposition \ref{p:propor6} which establishes that in the case of three degrees of freedom there are examples whose vector fields are not Lie symmetries of their corresponding KHK maps. These results can be found in Sections \ref{s:r2}, \ref{s:r4} and \ref{s:r6} respectively. In Section \ref{s:add}  we explore the symplecticity of the maps considered in this work. In the final section, we outline key conclusions and suggest potential directions for future research.

\section{Preliminaries and some definitions}\label{s:defi}
In this section we recall some definitions, as well as  result in \cite{Celledoni2013}, and we obtain and a direct consequence of this result.

A map $\Phi$ \emph{preserves a measure} that is absolutely continuous with respect to Lebesgue measure and has a non-vanishing density $\nu$ if $m(\Phi^{-1}(B)) = m(B)$ for any Lebesgue measurable set $B\in \U\subseteq \R^n$, where $ m(B)=\int_{B} \nu(x,y)\,d{\bf{x}}$ with $d{\bf{x}}=dx_1\wedge\cdots\wedge dx_n$, which implies that
\begin{equation}	\label{e:mesura}
\nu(\Phi({\bf{x}}))|{\mathrm D}\Phi({\bf{x}})|=\nu({\bf{x}}),
\end{equation}
where $|{\mathrm D}\Phi({\bf{x}})|$ is the determinant of the differential matrix ${\mathrm D}\Phi$.
In \cite{Celledoni2013} it is proved the following result:

\begin{proposition}[\cite{Celledoni2013}]\label{p:measure}
For all cubic Hamiltonian systems on symplectic vector spaces and on all Poisson vector spaces with constant Poisson structure,
the KHK maps have an invariant measure with density
$$
 \nu({\bf{x}})=\frac{1}{\left|I-\frac{1}{2}h {\mathrm D}X({\bf{x}})\right|}.
$$
\end{proposition}

Given a differentiable map $\Phi$
defined in an open set $\U\in\R^n$, a \emph{Lie symmetry} of $\Phi$ is a vector field $X$, defined in
$\U,$ such that $\Phi$ maps any orbit of the differential system
\begin{equation}\label{e:edo}
\dot{\bf{x}}=X({\bf{x}}),
\end{equation}
 into another orbit of the system. Equivalently, it is a {vector field}
such that the differential equation (\ref{e:edo}) is invariant by the change of variables $\mathbf{u}=\Phi(\mathbf{x})$. Such a vector field is characterized by the compatibility
equation
\begin{equation}\label{e:Lie-sym-char}
X_{|\Phi({\bf{x}})}={\mathrm D}\Phi({\bf{x}})\,X({\bf{x}}),
\end{equation} for $\bf{x}\in \U$,
where
$X_{|\Phi({\bf{x}})}$ means the vector field evaluated at $\Phi({\bf{x}})$. See  \cite{haggar}, for more details.

From a dynamical viewpoint, if $\Phi$ preserves a given  orbit $\gamma$ of $X$, there exists $\tau\geq 0$ such that $\Phi(p)=\varphi(\tau,p)$, for all
$p\in \gamma$, where $\varphi$  is the flow of $X$ (i.e. $\tau$ only depends on $\gamma$).  As a consequence of this fact, the action of $\Phi$ over the points in $\gamma$ is linear \cite[Theorem 1]{CGM08}.

According to Theorem 12 of \cite{CGM08}, we have that if $\Phi$ is \emph{a planar map} with a first integral $H$, then the vector field
$X(\mathbf{x})=\mu(\mathbf{x})\left(-H_y,H_x\right)$ is a Lie symmetry of $\Phi$ if and only if the map has an invariant measure with density $\nu=1/\mu$. As a consequence of this fact and of  Proposition \ref{p:measure}, we obtain that all the \emph{planar}  KHK maps that come from a discretization of cubic Hamiltonians have a Lie symmetry:

\begin{proposition}\label{p:propo3} Let $\Phi_h$ be a KHK map associated to a \emph{planar} Hamilonian vector field $X$ with a cubic Hamiltonian $H$. Then, the vector field
$$
Y(x,y)=\left|I-\frac{1}{2}h {\mathrm D}X(x,y))\right|\,\left(-\widetilde{H}_y,\widetilde{H}_x\right),
$$
where $\widetilde{H}$ is given by \eqref{e:modHam} is a Lie symmetry of $\Phi_h$.
\end{proposition}

\section{One-degree of freedom systems}\label{s:r2}
In this section, we consider the planar (one-degree of freedom) Hamiltonian vector field
$$
X=-\frac{\partial H}{\partial y} \frac{\partial}{\partial x}+\frac{\partial H}{\partial x}\frac{\partial}{\partial y},
$$
whose associated system is
\begin{equation}\label{e:ham2}
 \dot{x}=-\dfrac{\partial H}{\partial y}=-H_y, \qquad
\dot{y}=\dfrac{\partial H}{\partial x}=H_x,
\end{equation}
with  cubic Hamiltonian function
\begin{equation}\label{e:genham}
H(x,y)=\displaystyle{\sum\limits_{0< i+j\leq 3}} a_{ij}x^iy^j.
 \end{equation}
The main result of the section is the following:

\begin{theorem}\label{t:Hamiltonians de R2}
The Hamiltonian vector fields of $\R^2$ with Hamiltonian functions of degree at most three, given by \eqref{e:genham}, for which their associated KHK maps preserve the original Hamiltonian, are:
\begin{enumerate}[(a)]
    \item The vector field $X_1$, with associated system
    $$\dot{x}=-\left(a_{11} x +2 a_{02} y + a_{01}\right),\,\quad
    \dot{y}=\dfrac{a_{11}}{2 a_{02}}\left(a_{11} x +2 a_{02} y + a_{01}\right),
    $$
    and Hamiltonian
    $H_1(x,y)=\frac{1}{4 a_{02}}\left(a_{11} x +2 a_{02} y\right) \left(a_{11} x +2 a_{02} y+2 a_{01}\right),$ whose associated KHK map is $$\Phi_1(x,y)= \left(\begin{array}{c}
\left(-h a_{11}+1\right) x -2 h a_{02} y-h a_{01}
\\
 \frac{h a_{11}^{2} x}{2 a_{02}}+\left(h a_{11}+1\right) y +\frac{h a_{01} a_{11}}{2 a_{02}}
\end{array}\right).
$$
The functions $(H_{1})_{x}$ and $(H_{1})_{y}$ are also first integrals of $X_1$, but functionally dependent on~$H_1$.
        \item The vector field $X_2$, with associated system
            $$\dot{x}=0,\, \quad
    \dot{y}=3  a_{30} x^{2}+2 a_{20}x+a_{10},
    $$
    and Hamiltonian $H_2(x,y)=x \left(x^{2} a_{30}+a_{20} x+a_{10}\right)$ whose associated KHK map is
    $$\Phi_2(x,y)=\left(\begin{array}{c}
x
\\
 3 h \,x^{2} a_{30}+2 h a_{20} x+h a_{10}+y
\end{array}\right).
$$
 The function $(H_{2})_{x}$ is also a first integral of $X_2$, but functionally dependent on~$H_2$.

      \smallskip
    \item The vector field $X_3$, with associated system
    $$
    \dot{x}=-a_{01},\, \quad \dot{y}=a_{10},
    $$
    and Hamiltonian $H_3(x,y)=a_{10} x +a_{01} y$ whose associated KHK map is $
\Phi_3(x,y)=\left(\begin{array}{c}
x-h a_{01}
\\
 y+h a_{10}
\end{array}\right).
$

    \smallskip

        \item The vector field $X_4$, with associated system
    $$
\dot{x}=-3 a_{03} y^{2}-2 a_{02} y-a_{01},\, \quad \dot{y}=0,
    $$
    and Hamiltonian $H_4(x,y)=y \left(a_{03} y^{2}+a_{02} y+a_{01}\right)$ whose associated KHK map is $$
\Phi_4(x,y)=\left(\begin{array}{c}
-3 h \,a_{03} y^{2}-2 h a_{02} y-h a_{01}+x
\\
 y
\end{array}\right).
$$
 The function $(H_{4})_{y}$ is also a first integral of $X_4$, but functionally dependent on~$H_4$.

    \smallskip
\item The vector field $X_5$, with associated system
    $$
\dot{x}=-\frac{1}{3 a_{03} a_{12}}P(x,y),\, \quad
\dot{y}=\frac{1}{9 a_{03}^{2}}P(x,y),
    $$where $P(x,y)=a_{12}^{3}x^{2} +6 a_{03} a_{12}^{2}x y +9 a_{03}  a_{12}y^{2}+3  a_{03} a_{11} a_{12}x+9  a_{03}^{2} a_{11}y+9 a_{03}^{2} a_{10}$,
    and Hamiltonian \begin{align*}
    H_5(x,y)=&\frac{1}{54 a_{03}^{2} a_{12}}\left(a_{12} x +3  a_{03}y\right)\left(2  a_{12}^{3}x^{2}+12  a_{03} a_{12}^{2}x y+18 a_{03}  a_{12}y^{2}\right.\\
&\left.+9  a_{03} a_{11} a_{12}x+27  a_{03}^{2} a_{11}y+54 a_{03}^{2} a_{10}\right)
\end{align*}
whose associated KHK map is
$$
\Phi_5(x,y)=\left(\begin{array}{c}
-\frac{h a_{12}^{2} x^{2}}{3 a_{03}}+\left(-2 h  a_{12}y-h a_{11}+1\right) x -3 h \,a_{03} y^{2}-
\frac{3 h a_{11} a_{03} y}{a_{12}}-\frac{3 h a_{10} a_{03}}{a_{12}}
\\
 \frac{h a_{12}^{3} x^{2}}{9 a_{03}^{2}}+\left(\frac{2 h a_{12}^{2} y}{3 a_{03}}+\frac{h a_{11} a_{12}}{3 a_{03}}\right) x +h a_{12} y^{2}+\left(h a_{11}+1\right) y +
h a_{10}
\end{array}\right).
$$
The functions $(H_{5})_{x}$ and $(H_{5})_{y}$ are also first integrals of $X_5$ but functionally dependent on~$H_5$.
\end{enumerate}
Furthermore, each vector field $X_i$, $i=1,\ldots,5$ is a Lie symmetry of the corresponding map $\Phi_i$.

\end{theorem}

We emphasize that, according to Corollary \ref{c:simplecticr2} in Section \ref{s:add}, all the aforementioned maps $\Phi_i$, $i=1,\ldots,5$, are symplectic.

\begin{remark} Additionally, we point out the following:
\begin{itemize}
\item[(a)] For all the families of vector fields $X=(P,Q)$ of  Theorem \ref{t:Hamiltonians de R2} we have $\Phi_h=(x+hP,y+hQ)$.
\item[(b)] Under the change of coordinates $X=y, Y=x$ and $T= -t$ and renaming the constants $a_{30}\rightarrow a_{03}, a_{20}\rightarrow a_{02}, a_{10}\rightarrow a_{01}$ families $X_2$ and $X_4$ are the same.
\item[(c)] For each specific case outlined in Theorem \ref{t:Hamiltonians de R2}, it is possible to apply particular rescalings to eliminate some of the parameters $a_{ij}$, thereby simplifying the expressions. However, in this article, we have chosen not to perform these specific rescalings in order to facilitate comparisons between the different families by maintaining a consistent set of common parameters.

\end{itemize}

\end{remark}

To prove  Theorem \ref{t:Hamiltonians de R2}, we first establish the following characterization of the conditions under which a system of the form \eqref{e:ham2} has an associated KHK map that preserves the original Hamiltonian.

\begin{lemma}\label{l:invr2}
The KHK map \eqref{e:phi} associated to  a planar Hamiltonian system \eqref{e:ham2} preserves the original Hamiltonian function \eqref{e:genham} if and only if
\begin{equation}\label{e:condicioar2}
2H_xH_yH_{xy}-H_{xx}H_y^2-H_{yy}H_x^2=0.
\end{equation}
 \end{lemma}

  Observe that, by using the notation $\dot{g}=\{H,g\}$, where $\{\,\}$ is the usual Poisson bracket and $\dot{}=d/dt$, see \cite{Marsden1999}, the Equation \eqref{e:condicioar2} also writes as
$\dot{H}_xH_y-\dot{H}_yH_x=0.$

 \begin{proof}[Proof of Lemma \ref{l:invr2}]
Set ${\bf{x}}=(x,y)$ and $X({\bf{x}})=(-H_y, H_x)^T$, then we have
 $$
 I-\frac{1}{2}hDX({\bf{x}})=
 \left(\begin{array}{cc}
1+\frac{1}{2}hH_{xy} & \frac{1}{2}hH_{yy}\\
-\frac{1}{2}hH_{xx} & 1-\frac{1}{2}hH_{xy}
\end{array}
 \right).
 $$
Relation \eqref{e:restriction}  is
\begin{align*}
0&=
\nabla H({\bf{x}})^T \left(I-\frac{1}{2}hDX({\bf{x}})\right)^{-1}X({\bf{x}})\\
&= \frac{1}{A(x,y)}(H_x, H_y) \left(\begin{array}{cc}
 1-\frac{1}{2}hH_{xy} & -\frac{1}{2}hH_{xx}\\
 -\frac{1}{2}hH_{yy}&  1+\frac{1}{2}hH_{xy}
\end{array}
 \right)\left(\begin{array}{r}
 -H_y\\
 H_y
 \end{array}\right)=0,
\end{align*}
where $A(x,y)=1+\frac{1}{4}h^2\left(H_{xx} H_{yy}-H_{xy}^2\right)$, which yields to 
$$ 
2H_xH_yH_{xy}-H_{xx}H_y^2-H_{yy}H_x^2=\{H,H_x\}H_y-\{H,H_y\}H_x=\dot{H_x}H_y-\dot{H_y}H_x=0.
$$
\end{proof}

\begin{proof}[Proof of Theorem \ref{t:Hamiltonians de R2}]
By imposing the relation \eqref{e:condicioar2} on a Hamiltonian of the form \eqref{e:ham2}, and equating the coefficients with the aid of a computer algebra system, we obtain the following system of 21 equations that are symmetric with respect to the change $ a_{ij} \leftrightarrow a_{ji}$:  

\begin{align*}
&a_{01}^{2} a_{20}-a_{01} a_{10} a_{11}+a_{02} a_{10}^{2}=0,\\ 
& 2 a_{02} a_{12} a_{30}+3 a_{03} a_{11} a_{30}+2 a_{03} a_{20} a_{21}- a_{11} a_{12} a_{21}
=0,\\ 
& 12 a_{01} a_{02} a_{30}-2 a_{01} a_{11} a_{21}+4 a_{02} a_{11} a_{20}+12 a_{03} a_{10} a_{20}-2 a_{10} a_{11} a_{12}-a_{11}^{3}=0,\\
& 3 a_{01}^{2} a_{30}-2 a_{01} a_{10} a_{21}+4 a_{02} a_{10} a_{20}+a_{10}^{2} a_{12}-a_{10} a_{11}^{2}
=0,\\ 
& 3 a_{03} a_{10}^{2}-2 a_{01} a_{10} a_{12}+4 a_{01} a_{02} a_{20}+a_{01}^{2} a_{21}-a
_{01} a_{11}^{2}=0,\\ 
&  a_{30} \left(3 a_{12} a_{30}-a_{21}^{2}\right)=0,\\ 
& a_{03} \left(3 a_{03} a_{21}-a_{12}^{2}\right)=0,\\ 
& 
4 a_{02} a_{03} a_{21}- a_{02} a_{12}^{2}+3 a_{03}^{2} a_{20}- a_{03} a_{11} a_{12}
=0,\\ 
& 4 a_{12} a_{20} a_{30}- a_{20} a_{21}^{2}+3 a_{02} a_{30}^{2}- a_{11} a_{21} a_{30}
=0,\\ 
& 3 a_{01} a_{11} a_{30}-2 a_{01} a_{20} a_{21}+6 a_{02} a_{10} a_{30}+4 a_{02} a_{20}^{2}-3 a_{10} a_{11} a_{21}+4 a_{10} a_{12} a_{20}-a_{11}^{2} a_{20}
=0,\\ 
&3 a_{03} a_{10} a_{11}-2 a_{02} a_{10} a_{12}+6 a_{01} a_{03} a_{20}+4 a_{02}^{2} a_{20}-3 a_{01} a_{11} a_{12}+ 4 a_{01} a_{02} a_{21}-a_{02} a_{11}^{2}
=0,\\ 
& 6 a_{02} a_{20} a_{30}+3 a_{10} a_{12} a_{30}- a_{10} a_{21}^{2}-2 a_{11} a_{20} a_{21}+2 a_{12} a_{20}^{2}
=0,\\ 
& 6 a_{02} a_{03} a_{20}+3 a_{01} a_{03} a_{21}- a_{01} a_{12}^{2}-2 a_{02} a_{11} a_{12}+2 a_{02}^{2} a_{21}
=0,\\ 
& 9 a_{03} a_{30}^{2}+2 a_{12} a_{21} a_{30}- a_{21}^{3}=0,\\ 
& 9 a_{03}^{2} a_{30}+2 a_{03} a_{12} a_{21}- a_{12}^{3}=0,\\ 
&6 a_{03} a_{21} a_{30}+ a_{12}^{2} a_{30}- a_{12} a_{21}^{2}
=0,\\ 
& 6 a_{03} a_{12} a_{30}+ a_{03} a_{21}^{2}- a_{12}^{2} a_{21}
=0,\\ 
& 6 a_{02} a_{03} a_{30}+ a_{03} a_{11} a_{21}+2 a_{03} a_{12} a_{20}- a_{11} a_{12}^{2}
=0,\\ 
& 6 a_{03} a_{20} a_{30}+ a_{11} a_{12} a_{30}+2 a_{02} a_{21} a_{30}- a_{11} a_{21}^{2}
=0,\\ 
& 3 a_{01} a_{12} a_{30}- a_{01} a_{21}^{2}+6 a_{02} a_{11} a_{30}+2 a_{02} a_{20} a_{21}+9 a_{03} a_{10} a_{30}+6 a_{03} a_{20}^{2}- a_{10} a_{12} a_{21}\\
&-2 a_{11}^{2} a_{21}
=0,\\
&3 a_{03} a_{10} a_{21}- a_{10} a_{12}^{2}+6 a_{03} a_{11} a_{20}
+2 a_{02} a_{12} a_{20}+ 9 a_{01} a_{03} a_{30}+6 a_{02}^{2} a_{30}- a_{01} a_{12} a_{21}\\
&-2 a_{11}^{2} a_{12}
=0.
\end{align*}

Using again the assistance of a symbolic computing software, we obtain the solutions of this system, obtaining that the only solutions are those that give rise to the Hamiltonians $H_i$, $i=1,\ldots,5$, and its associated vector fields $X_i$ and differential systems in the statement.

We can find the KHK maps associated with the corresponding Hamiltonian vector fields by using \eqref{e:phi}.  Finally, by using the compatibility equation \eqref{e:Lie-sym-char}, we find that $X_{i|\Phi_i(\bf{x})}=D\Phi_i({\bf{x}})\,X_i({\bf{x}})$ for $i=1,\ldots,5$, so each vector field $X_i$ is a Lie symmetry of the maps $\Phi_i$.

\end{proof}

The dynamical and the algebraic geometric properties of the KHK maps associated with \emph{generic} planar Hamiltonian vector fields with cubic Hamiltonians  have been studied in  \cite{Petrera19-1,Petrera19-2}. 
Remember that the space \(\mathcal{H}_3\) of planar vector fields with cubic Hamiltonians is homeomorphic to \(\mathbb{R}^9\) in the topology of the coefficients \cite[p. 202]{ADL}. A subset  of this space is \emph{generic} if it is open and dense. As a consequence of Theorem \ref{t:teoham}, the KHK maps associated with Hamiltonian vector fields preserve the cubic first integral $\tilde{H}$ defined in \eqref{e:modHam}. For the generic cases, the energy levels of this Hamiltonian $\tilde{H}$ are, except perhaps for a finite set of levels, elliptic curves. Thus, the KHK maps associated with cubic Hamiltonians are, generically, birational integrable maps preserving genus-$1$ fibrations, and therefore can be described in terms of the linear action on the group structure of the preserved elliptic curves \cite{D,JRV}.
We notice, however, that the cases presented in Theorem \ref{t:Hamiltonians de R2} are, obviously, non-generic within the topological space of vector fields with cubic Hamiltonians, since equation \eqref{e:restriction} must be satisfied. More specifically:

\begin{remark}\label{r:remark7}
The cases obtained in Theorem \ref{t:Hamiltonians de R2} correspond to instances where the pencil of elliptic curves associated with the Hamiltonian is always factorizable. The structure of the singular fibers of the elliptic fibrations associated with the modified Hamiltonians has been studied in \cite[Corollary 2.12]{GMQ24}. The families described in Theorem 5 are singular cases of those characterized in that reference.
\end{remark}

As noticed in the above Remark there is no planar quadratic Hamiltonian system satisfying condition \eqref{e:restriction} (or equivalently, \eqref{e:condicioar2}) with  cubic irreducible Hamiltonian function. Also, no one of the associated vector fields have coprime components. A straightforward analysis indicates that in all the cases in Theorem \ref{t:Hamiltonians de R2}, the energy level sets are given by parallel straight lines, some of them (at most two) full of singular points of the associated differential systems.
From a qualitative point of view, the dynamics of these fields is trivial, because either they are linear or, by reparameterizing the time, they give rise to systems with constant components.
The dynamics of the associated KHK maps is, therefore, not so rich as the ones that are displayed in the general case studied, for instance in \cite{Celledoni2019b,Kam19,Petrera19-1,Petrera19-2}. By construction, the obtained KHK maps will preserve the above mentioned invariant straight lines. This is also a consequence of the fact that, in general, the KHK maps
preserve the affine Darboux polynomials of any quadratic ODE \cite[Theorem 1]{Celledoni2019c} as well as, notably, the Runge-Kutta methods do, \cite[Theorem 3.1]{Tapley} (see also \cite{Celledoni2022} for more information on using Darboux polynomials for the study of KHK maps). Let us give an example:

\medskip

\noindent \textbf{Example A.}  Consider the following Hamiltonian that belongs to the class of systems considered in statement (e) of Theorem \ref{t:Hamiltonians de R2}:  \begin{align}
H_5(x,y)&=\frac{1}{27} x^{3}+\frac{1}{3} x^{2} y +x \,y^{2}+y^{3}-\frac{1}{6} x^{2}-x y -\frac{3}{2} y^{2}-x -3 y\label{e:ex}\\
&=\frac{1}{27}  \left(x+3 y  \right)\left(x +3 y -\frac{9}{4}+\frac{3 \sqrt{57}}{4}\right)\left(x +3 y -\frac{9}{4}-\frac{3 \sqrt{57}}{4}\right).\nonumber
\end{align}
The Hamiltonian vector field has associated differential  system $\{
\dot{x}=-\frac{1}{3} P(x,y),\, \quad \dot{y}=\frac{1}{9} P(x,y)\}$,
with $P(x,y)=x^{2}+6 x y +9 y^{2}-3 x -9 y -9$.
The phase portrait of the system is, therefore, very simple: all the orbits lie in straight lines of the form $y=-x/3+c$. The lines $y=-x/3+(1\pm\sqrt{5})/2$ are filled by singular points.
Any orbit with initial condition $(x_0,y_0)$ such that
$-x_0/3+(1-\sqrt{5})/2<y_0<-x_0/3+(1+\sqrt{5})/2$
evolves to the right and down through a line of the form $y=-x/3+c$. The rest of orbits, not in the singular lines, evolve to the left and up,  see Figure \ref{f:f1}.

The associated KHK map
\begin{align*}
\Phi(x,y)=&\left(-\frac{h \,x^{2}}{3}-2 h x y +\left(h +1\right) x -3 h \,y^{2}+3 h y +3 h,\right.\\
&\left.
\frac{h \,x^{2}}{9}+\frac{2 h x y}{3}-\frac{h x}{3}+h \,y^{2}+\left(-h +1\right) y -h\right),
\end{align*}
captures these behaviors.

\begin{figure}[H]
\centerline{
\includegraphics[scale=0.35]{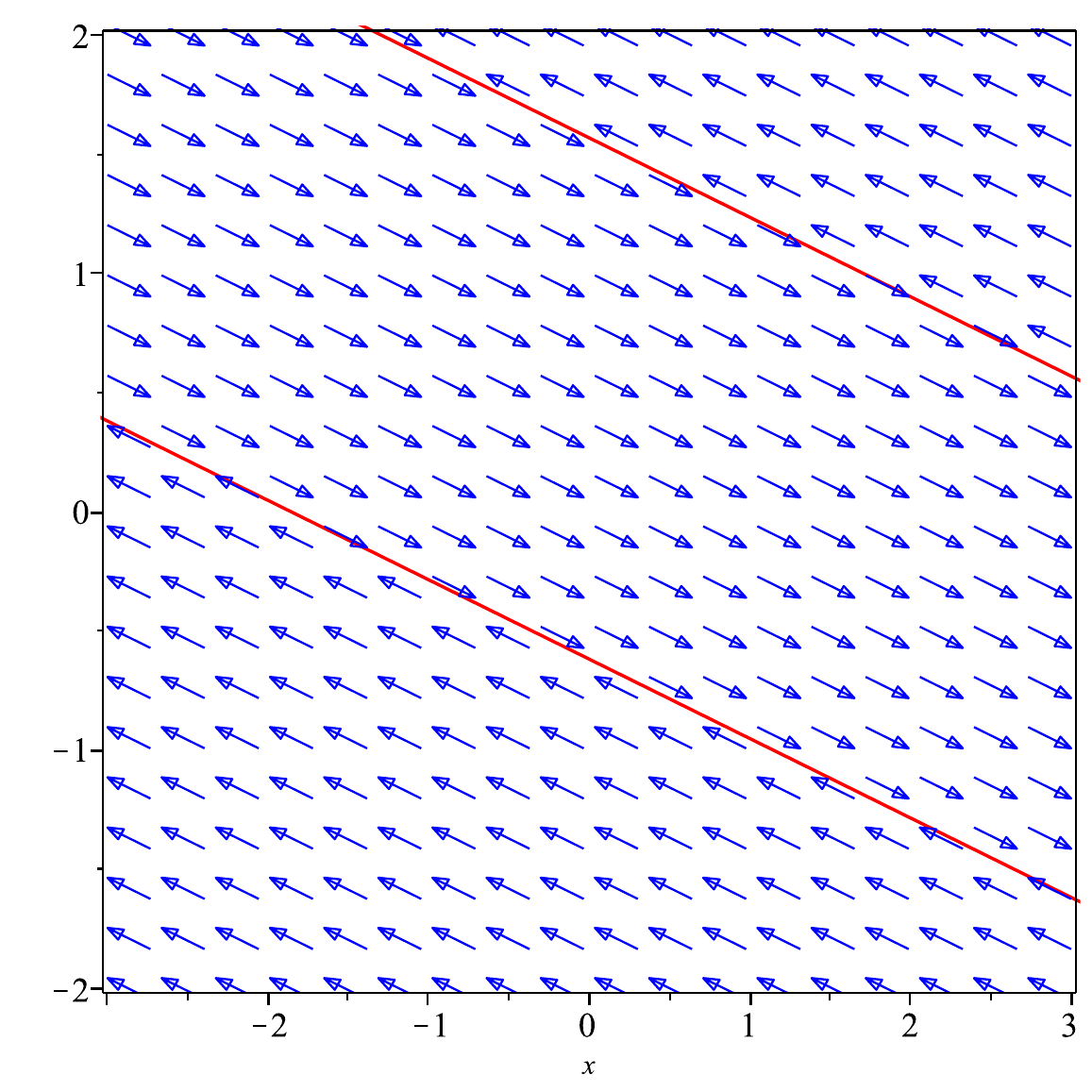}\includegraphics[scale=0.35]{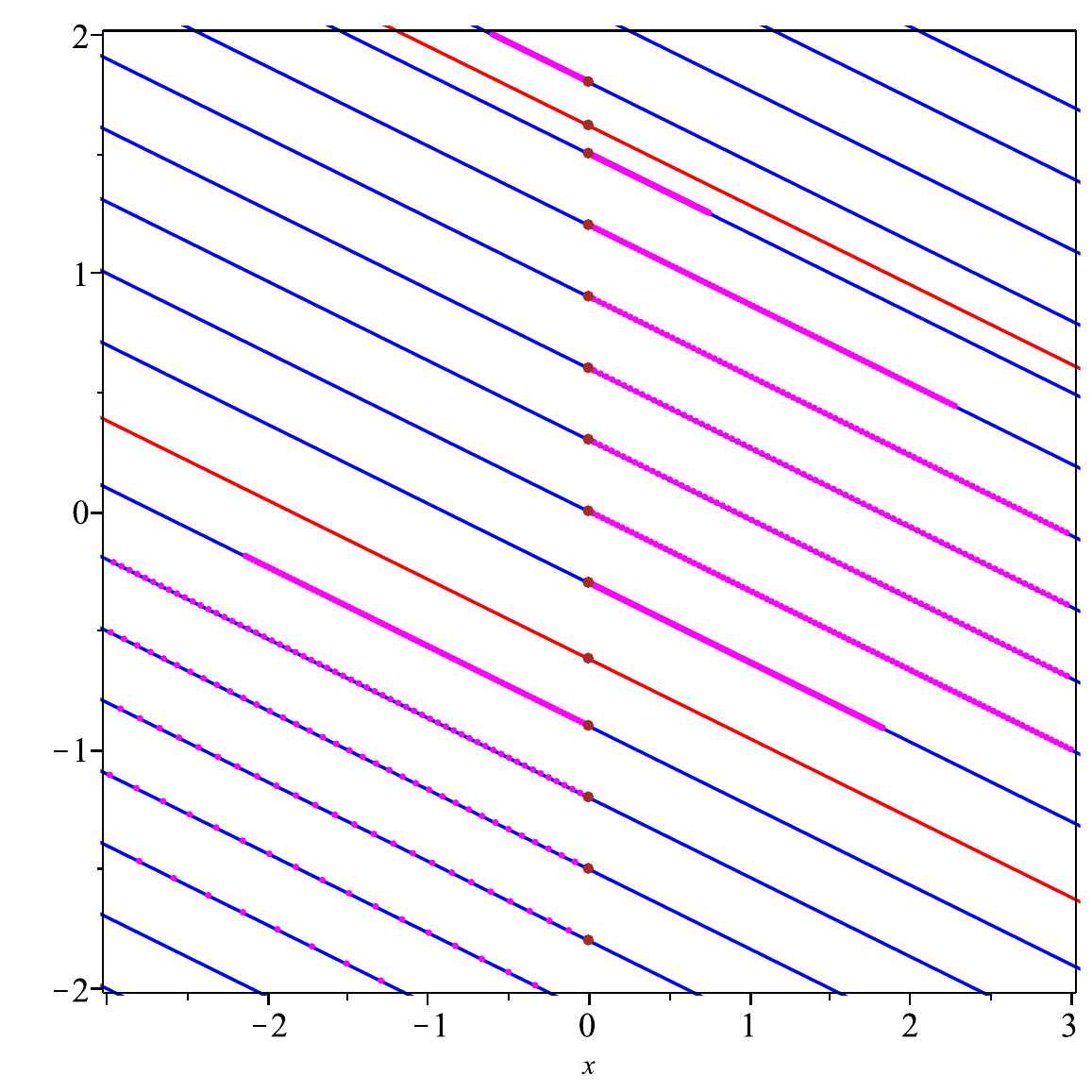}}\caption{
(Left) Vector field associated to the Hamiltonian \eqref{e:ex}. The lines  $y=-x/3+(1\pm\sqrt{5})/2$ are full of singular points, in red. (Right)  Orbits of the Hamiltonian system (in blue). There are also shown 100 iterates of some orbits of the map $\Phi_h$ with $h=0.01$ (in magenta), with initial conditions in brown. The lines filled by singular points of the vector field are also filled by fixed points of the map.}\label{f:f1}
\end{figure}

\section{Two-degrees of freedom Hamiltonian systems}\label{s:r4}

We consider Hamiltonian vector fields with two degrees of freedom
$$
X=-\frac{\partial H}{\partial y} \frac{\partial}{\partial x}+\frac{\partial H}{\partial x}\frac{\partial}{\partial y}
-\frac{\partial H}{\partial w} \frac{\partial}{\partial z}+\frac{\partial H}{\partial z}\frac{\partial}{\partial w},
$$ with cubic Hamiltonian function
 \begin{equation}\label{e:genham4}
H(x,y,z,w)=\displaystyle{\sum\limits_{0< i+j+k+l\leq 3}
} a_{ijkl}x^iy^jz^kw^l.
 \end{equation}
 The associated two degrees of freedom Hamiltonian system is 

\begin{equation}\label{e:ham4}
 \dot{x}=-\dfrac{\partial H}{\partial y}=-H_y, \quad
\dot{y}=\dfrac{\partial H}{\partial x}=H_x,\quad
\dot{z}=-\dfrac{\partial H}{\partial w}=-H_w, \quad
\dot{w}=\dfrac{\partial H}{\partial z}=H_z.
\end{equation}

These fields can also be written in the form $q_i=\partial H/\partial p_i$, $p_i=-\partial H/\partial q_i$, with $q_1=x,q_2=z$ and $p_1=y,p_2=w$. The relation between the associated KHK maps $\Psi_h$ of these vector fields and $\Phi_h$, the ones associated with the fields in Equation \eqref{e:ham4}, is given by
$
{\Psi}_h(q_1,q_2,p_1,p_2)=\left(\Phi _{1,-h},\Phi _{3,-h},\Phi _{2,-h},\Phi _{4,-h}\right)(q_1,p_1,q_2,p_2).
$ See Lemma \ref{l:noulemanou} in  Appendix \ref{s:appendixnou}.

The following result gives a characterization the conditions for which a system of the form \eqref{e:ham4} has an associated KHK map that preserves the original Hamiltonian.

\begin{lemma}\label{l:invr4}
The KHK map \eqref{e:phi} associated with a  Hamiltonian system \eqref{e:ham4} preserves the original Hamiltonian function \eqref{e:genham4} if and only if
\begin{equation}\label{e:condicioar4}
\begin{array}{l}
H_x\{H,H_y\}-H_y\{H,H_x\}+\{H,H_w\}H_z-\{H,H_z\}H_w=0 \mbox{ and}\\
\\
H_x\left(\{H_,H_y\}\{H_z,H_w\}-\{H_,H_z\}\{H_y,H_w\}+\{H_,H_w\}\{H_y,H_z\}\right)\\
-H_y\left(\{H_,H_x\}\{H_z,H_w\}-\{H_,H_z\}\{H_x,H_w\}+\{H_,H_w\}\{H_x,H_z\}\right)\\
+H_z\left(\{H_,H_x\}\{H_y,H_w\}-\{H_,H_y\}\{H_x,H_w\}+\{H_,H_w\}\{H_x,H_y\}\right)\\
-H_w\left(\{H_,H_x\}\{H_y,H_z\}-\{H_,H_y\}\{H_x,H_z\}+\{H_,H_z\}\{H_x,H_y\}\right)=0.
\end{array}
\end{equation}
 \end{lemma}

The first equation in \eqref{e:condicioar4} also writes as
$ \dot{H}_xH_y-\dot{H}_yHx+\dot{H}_zH_w-\dot{H}_wHz=0$, where  $\dot{g}=\{H,g\}$.

\begin{proof}[Proof of Lemma \ref{l:invr4}]
In this case, the relation \eqref{e:restriction} gives
\begin{equation}\label{e:h0h2r4}
\nabla H({\bf{x}})^T \left(I-\frac{1}{2}h{\mathrm D}X({\bf{x}})\right)^{-1}X({\bf{x}})=\frac{h(\Lambda_0+\Lambda_2 h^2)}{\left|I-\frac{1}{2}h {\mathrm D}X({\bf{x}})\right|}=0,
\end{equation}
with

\begin{align*}
\Lambda_0=&8 \left(H_{w}^{2} H_{zz} -2 H_{w} H_{x} H_{yz} +2 H_{w} H_{xz} H_{y} -2 H_{w} H_{z} H_{zw} + H_{ww} \,H_{z}^{2}+ H_{x}^{2} H_{yy}\right.\\
&\left. -2 H_{x} H_{xy} H_{y} +2 H_{x} H_{yw} H_{z} -2 H_{xw} H_{y} H_{z} + H_{xx} \,H_{y}^{2}\right),
\end{align*}
and
\begin{align*}
\Lambda_2=&
2\left( H_{x}^{2} H_{yy} H_{zz} H_{ww} - H_{x}^{2} H_{yy} \,H_{zw}^{2}- H_{x}^{2} H_{yz}^{2} H_{ww} +2 H_{x}^{2} H_{yz} H_{yw} H_{zw} - H_{x}^{2} H_{yw}^{2} H_{zz}\right.\\
& -2 H_{x} H_{y} H_{xy} H_{zz} H_{ww} +2 H_{x} H_{y} H_{xy} \,H_{zw}^{2}+2 H_{x} H_{y} H_{xz} H_{yz} H_{ww} -2 H_{x} H_{y} H_{xz} H_{yw} H_{zw}\\
& -2 H_{xw} H_{x} H_{y} H_{yz} H_{zw} +2 H_{xw} H_{x} H_{y} H_{yw} H_{zz} +2 H_{x} H_{z} H_{xy} H_{yz} H_{ww} \\
& -2 H_{x} H_{z} H_{xy} H_{yw} H_{zw}-2 H_{x} H_{z} H_{xz} H_{yy} H_{ww} +2 H_{x} H_{z} H_{xz} \,H_{yw}^{2}+2 H_{xw} H_{x} H_{z} H_{yy} H_{zw} \\
&-2 H_{xw} H_{x} H_{z} H_{yz} H_{yw} -2 H_{x} H_{w} H_{xy} H_{yz} H_{zw} +2 H_{x} H_{w} H_{xy} H_{yw} H_{zz}   \\
&+2 H_{x} H_{w} H_{xz} H_{yy} H_{zw}-2 H_{x} H_{w} H_{xz} H_{yz} H_{yw}-2 H_{xw} H_{x} H_{w} H_{yy} H_{zz} +2 H_{xw} H_{x} H_{w} H_{yz}^{2}\\
&+ H_{y}^{2} H_{xx} H_{zz} H_{ww} - H_{y}^{2} H_{xx} H_{zw}^{2}- H_{y}^{2} H_{xz}^{2} H_{ww} +2 H_{xw} \,H_{y}^{2} H_{xz} H_{zw} - H_{xw}^{2} H_{y}^{2} H_{zz} \\
&-2 H_{y} H_{z} H_{xx} H_{yz} H_{ww} +2 H_{y} H_{z} H_{xx} H_{yw} H_{zw} +2 H_{y} H_{z} H_{xy} H_{xz} H_{ww}  \\
&-2 H_{xw} H_{y} H_{z} H_{xy} H_{zw} -2 H_{xw} H_{y} H_{z} H_{xz} H_{yw}+2 H_{xw}^{2} H_{y} H_{z} H_{yz} +2 H_{y} H_{w} H_{xx} H_{yz} H_{zw}  \\
&-2 H_{y} H_{w} H_{xx} H_{yw} H_{zz}-2 H_{y} H_{w} H_{xy} H_{xz} H_{zw}+2 H_{xw} H_{y} H_{w} H_{xy} H_{zz} +2 H_{y} H_{w} \,H_{xz}^{2} H_{yw} \\
& -2 H_{xw} H_{y} H_{w} H_{xz} H_{yz} + H_{z}^{2} H_{xx} H_{yy} H_{ww}- H_{z}^{2} H_{xx} \,H_{yw}^{2}- H_{z}^{2} H_{xy}^{2} H_{ww} \\
&+2 H_{xw} \,H_{z}^{2} H_{xy} H_{yw} - H_{xw}^{2} H_{z}^{2} H_{yy}-2 H_{z} H_{w} H_{xx} H_{yy} H_{zw} +2 H_{z} H_{w} H_{xx} H_{yz} H_{yw} \\
 &+2 H_{z} H_{w} H_{xy}^{2} H_{zw} -2 H_{z} H_{w} H_{xy} H_{xz} H_{yw}-2 H_{xw} H_{z} H_{w} H_{xy} H_{yz} +2 H_{xw} H_{z} H_{w} H_{xz} H_{yy} \\
 &\left.+ H_{w}^{2} H_{xx} H_{yy} H_{zz} - H_{w}^{2} H_{xx} H_{yz}^{2}- H_{w}^{2} H_{xy}^{2} H_{zz} +2 H_{w}^{2} H_{xy} H_{xz} H_{yz} - H_{w}^{2} H_{xz}^{2} H_{yy}\right).
\end{align*}

A computation shows that
\begin{align*}
  \frac{\Lambda_0}{8}&=\{H,H_x\}H_y-\{H,H_y\}H_x+\{H,H_z\}H_w-\{H,H_w\}H_z\\
  &=\dot{H}_xH_y-\dot{H}_yH_x+\dot{H}_zH_w-\dot{H}_wH_z,
\end{align*}
where $\dot{}=d/dt$, so we obtain the first equation in \eqref{e:condicioar4}.
Analogously, an involved but straightforward computation shows that
  $\Lambda_2/2$ is the left hand side of the second equation.~\end{proof}
  
The following result summarizes the information we have obtained by solving equation \eqref{e:condicioar4} for a Hamiltonian of the form \eqref{e:genham4}.
Our calculations rely on the use of computer algebra software (Maple, in our case). In Appendix \ref{s:appendixa} we provide 54 non-trivial families of Hamiltonians preserved by the associated KHK maps. Unfortunately, we cannot guarantee that these solutions are the only ones, nor that they correspond to functionally independent cases.

  \begin{theorem}\label{t:Hamiltonians de R4} The following statements hold:
  \begin{enumerate}[(a)]
  \item  There are 54 families of Hamiltonian vector fields $X_i$ for $i = 1, \dots, 54 $ in $\mathbb{R}^4$, with Hamiltonian functions of degree at most three, as given by \eqref{e:genham4}, for which their associated KHK maps $\Phi_{i,h}$ preserve the original Hamiltonian. The corresponding Hamiltonians $H_i$ are listed in Appendix \ref{s:appendixa}.
  \item All vector fields corresponding to these Hamiltonians are Lie symmetries of the associated KHK maps.
\item     Some of the vector fields in the list, with Hamiltonian $H$, admit at least one additional first integral from the set $\{ H_x, H_y, H_z, H_w \}$, with only one being functionally independent of $H$. These vector fields are presented in Tables \ref{t:table1} and \ref{t:table2} below.
In all cases, the corresponding KHK maps  also preserve the same additional first integrals.
  \item All Hamiltonian vector fields in the list, with Hamiltonian $H$, that possess another first integral from the set $ \{H_x, H_y, H_z, H_w\}$, which is functionally independent of $H$, commute with the Hamiltonian vector field associated with the additional first integral.
  \end{enumerate}
  \end{theorem}

All the KHK maps associated with the vector fields referred to in the previous result are symplectic in the sense that they satisfy Equation \eqref{e:equacioambB} from Section \ref{s:add}. See Proposition \ref{p:totesymplectiquesr4} for further details.

  \begin{proof}
(a) By imposing that a Hamiltonian of the form \eqref{e:genham4} satisfies the condition \eqref{e:condicioar4}, we find that the expressions $ \Lambda_0$ and $\Lambda_2$, which appear in Equation \eqref{e:h0h2r4}, are polynomials that must be identically zero. These polynomials have degrees 5 and 7 in the variables $x, y, z, w$, and contain 126 and 330 coefficients, respectively.

Setting them to zero, we obtain a system of $456$ equations in the coefficients. By solving this system using the computer algebra software Maple on a laptop, we find $54$  solutions, which correspond to  $54$ Hamiltonians presented in 
Appendix~\ref{s:appendixa}.

By directly inspecting the vector fields $X_i$ associated with the corresponding Hamiltonians $H_i$ $i=1,\ldots,54$, we compile the information presented in the first four columns of Table 1. Statements (b), (c), and (d) are derived through straightforward computations, which can be carried out using a computer algebra system.~
\end{proof}

\begin{center}
\begin{table}[H]
\begin{tabular}{|l|l|l|l|l|l|l|l|}
\hline
$\exists$ null components & $\exists$ common factors & Linear systems & True quadratic \\
\hline
$X_{1},X_{6},X_{8},X_{9},$ &  $X_{1},X_{8},X_{9},X_{12},$ &  $X_{4},X_{5},X_{7},X_{17},$&  $X_{1},X_{2},X_{3},X_{6},$\\

$X_{10},X_{11},X_{12},X_{13},$ & $X_{13},X_{14},X_{20},X_{22},$ & $X_{21},X_{23},X_{24},X_{25},$ &
$X_{8},X_{9},X_{10},X_{11},$\\

$X_{14},X_{15},X_{20},X_{37},$ & $X_{28},X_{32},X_{38},X_{39},$ & $X_{26},X_{27},X_{29},X_{32},$ &
$X_{12},X_{13},X_{14},X_{15},$\\

$X_{42},X_{51}$ & $X_{42},X_{44},X_{45},X_{46},$ & $X_{33},X_{34},X_{35},X_{36},$ &
$X_{16},X_{18},X_{19},X_{20},$\\

 & $X_{47},X_{49},X_{50},X_{51},$ & $X_{40},X_{41},X_{42},X_{43},$ &
$X_{22},X_{28},X_{30},X_{31},$\\

 & $X_{53}$ & $X_{44},X_{48},X_{49},X_{52},$ &
$X_{37},X_{38},X_{39},X_{45},$\\

 & &$X_{54}$ &$X_{46},X_{47},X_{50},X_{53}$\\

 & & & \\

  \hline
  \hline
  First integral $H_x$ & First integral  $H_y$ & First integral $H_z$ & First integral $H_w$\\

  \hline
$X_{1},X_{6},X_{8},X_{9},X_{10},$
 & $X_{3},X_{8},X_{9},X_{10},X_{12},$&
$X_{1},X_{8},X_{9},X_{11},X_{12},$
&  $X_{1},X_{6},X_{15},X_{20},X_{22},$ \\

$X_{11},X_{12},X_{13},X_{14},$
&$X_{13},X_{14},X_{22},X_{25},$
&$X_{13},X_{14},X_{15},X_{20},$
& $X_{25}, X_{28},X_{30},X_{31},$  \\

$X_{20},X_{25},X_{38},X_{39},$
& $X_{28},X_{32}, X_{34},X_{37},$
&$X_{25},X_{32},X_{38},X_{39},$
&$ X_{37},X_{38},X_{39},X_{42},$ \\

$X_{42},X_{44},X_{45},X_{46},$
& $X_{38},X_{39}, X_{42},X_{44},$
&$X_{44},X_{45},X_{46},X_{47},$
&  $X_{44},X_{45},X_{46},X_{47},$\\

$X_{47}, X_{49},X_{50},X_{53}$
& $X_{45},X_{46},X_{47},X_{49},$
&$X_{49},X_{50},X_{51},X_{53}$
&$X_{49},X_{50},X_{51},X_{53}$ \\

&$X_{50},X_{51},X_{53}$
& &
\\
  \hline
\end{tabular}
\caption{Collected information about the vector fields $X_i$ associated to the Hamiltonians $H_i$ in  Appendix~\ref{s:appendixa}.}\label{t:table1}
\end{table}
\end{center}

\begin{sidewaystable}
{\small
 $$
\begin{array}{|c|c|c|c|c|c|c|c|}
\hline
  X \quad  & \mbox{First Integrals}  & \mbox{Rank}& c &X \quad  & \mbox{First Integrals}  & \mbox{Rank}& c\\
  \hline
  \hline
  X_1& H, H_x, H_z=cH_w & 2&  c=\frac{a_{1011}}{2 a_{1002}}&
  X_2 & H  & 1&\\
  X_3 & H,  H_y & 2& &
  X_4 & H  & 1&\\
 X_5 & H & 1 & &
 X_6 & H, H_x,  H_w & 2&\\
 X_7 & H  & 1 & &
 X_8 & H, H_x=cH_y, H_z & 2&c=\frac{a_{1100}}{2 a_{0200}}\\
 X_9 & H, H_x=cH_y, H_z & 2& c=\frac{a_{1110}}{2 a_{0210}}&
 X_{10} & H, H_x, H_y & 2 &\\
 X_{11} & H, H_x,  H_z & 2& &
 X_{12} & H, H_x=cH_y, H_z & 2&c=\frac{a_{1110}}{2 a_{0210}}\\
 X_{13} & H, H_x=cH_y, H_z & 2&c=\frac{a_{1100}}{2 a_{0200}} &
 X_{14} & H, H_x=cH_y, H_z & 2&c=\frac{a_{1110}}{2 a_{0210}}\\
 X_{15} & H,  H_z, H_w & 2& &
 X_{16} & H & 1 &\\
X_{17} & H & 1 & &
X_{18} & H & 1&\\
X_{19} & H & 1& &
X_{20} & H, H_x,  H_z=cH_w & 2&c=\frac{a_{0011}}{2 a_{0002}}\\
X_{21} & H  & 1& &
X_{22} & H,  H_y=cH_w & 2 &c=\frac{a_{0110}}{a_{0011}}\\
X_{23} & H & 1& &
X_{24} & H & 1&\\
X_{25} & H, H_x, H_y, H_z, H_w & 2& &
X_{26} & H & 1&\\
X_{27} & H & 1& &
X_{28} & H, H_y=cH_w & 2&c=\frac{a_{0100}}{a_{0001}}\\
X_{29} & H & 1& &
X_{30} & H, H_w & 2&\\
X_{31} & H, H_w & 2& &
X_{32} & H, H_y=cH_z, & 2&c=\frac{a_{0101}}{a_{1100}}\\
X_{33} & H & 1& &
X_{34} & H, H_y  & 2&\\
X_{35} & H & 1&  &
X_{36} & H,  & 1&\\
X_{37} & H, H_y,  H_w & 2&  &
X_{38} & H, H_x=c_1H_z, H_y=c_2Hw  & 2&c_1=-\frac{a_{0101}}{2 a_{0200}}, \, c_2=\frac{2 a_{0200}}{a_{0101}} \\
X_{39} & H, H_x=c_1H_z, H_y=c_2H_w & 2&c_1=-\frac{a_{0101}}{2 a_{0200}}, \, c_2=\frac{2 a_{0200}}{a_{0101}}  &
X_{40} & H & 1&\\
X_{41} & H & 1&  &
X_{42} & H, H_x=H_y,  H_w & 2&\\
X_{43} & H & 1&  &
X_{44} & H, H_x, H_y=cH_w, H_z & 2&c=\frac{2 a_{0200}}{a_{0101}}\\
X_{45} & H, H_x, H_y=cH_w, H_z & 2&c=-\frac{a_{1020}}{a_{2010}}&
X_{46} & H, H_x, H_y=cH_w, H_z & 2&c=-\frac{a_{1020}}{a_{2010}}\\
X_{47} & H, H_x, H_y=cH_w, H_z & 2&c=\frac{2 a_{0200}}{a_{0101}} &
X_{48} & H & 1&\\
X_{49} & H, H_x, H_y=cH_w, H_z & 2&c=-\frac{a_{1010}}{2 a_{2000}}&
X_{50} & H, H_x, H_y=cH_w, H_z & 2&c=-\frac{a_{1010}}{2 a_{2000}}\\
X_{51} & H,  H_y, H_z=cH_w & 2&c=\frac{a_{0110}}{a_{0101}}  &
X_{52} & H,  & 1&\\
X_{53} & H, H_x, H_y=cH_w, H_z & 2&c=\frac{2 a_{0200}}{a_{0101}}  &
X_{54} & H & 1&\\
\hline
\end{array}
$$}
\caption{Additional first integrals for the systems $X_i$ associated to the Hamiltonians $H_i$ in Appendix \ref{s:appendixa}.}\label{t:table2}
\end{sidewaystable}

\vfill
\newpage

We now present a few examples.

\medskip

\noindent \textbf{Example B.}  We consider  a particular case of the Hamiltonian of type $H_3$, given by
\begin{align*}
H(x,y,z,w)=&x^{3}-2 x^{2} z -2 x^{2} w +x \,z^{2}+2 x z w +x \,w^{2}-2 x^{2}-x y +2 x z \\
&+x w +y z +y w +z w +w^{2}+x +y +w\\
=& \left(-x +z +w+1\right) \left(-x^{2}+x z+x w   +x +y+w \right).
\end{align*}
The Hamiltonian vector field $X$ is given by the differential system:
\begin{align}
\dot{x}&=x -z -w -1,\nonumber\\
\dot{y}&=3 x^{2}-4 x z -4 x w +z^{2}+2 z w +w^{2}-4 x -y +2 z +w +1,\label{e:campexemple4}\\
\dot{z}&=2 x^{2}-2 x z -2 x w -x -y -z -2 w -1,\nonumber\\
\dot{w}&=-2 x^{2}+2 x z +2 x w +2 x +y +w.\nonumber
\end{align}
The associated KHK map of the vector field \eqref{e:campexemple4} is:
$\Phi_{h}(x,y,z,w)=(\Phi_1,\Phi_2,\Phi_3,\Phi_4)$,
where

\begin{align*}
\Phi_1=&\left(h +1\right) x -w h -z h -h,\\
\Phi_2=&\frac{h \left(h +6\right) x^{2}}{2}-h \left(h +4\right) x z -h \left(h +4\right) x w -h \left(h +4\right) x +\left(1-h \right) y +\frac{h \left(h +2\right) z^{2}}{2}\\
&+h \left(h +2\right) z w +h \left(h +2\right) z +\frac{h \left(h +2\right) w^{2}}{2}+h \left(h +1\right) w +\frac{h \left(h +2\right)}{2},
\\
\Phi_3=&\frac{h \left(h +4\right) x^{2}}{2}-h \left(h +2\right) x z -h \left(h +2\right) x w -h \left(h +1\right) x -h y +\frac{h^{2} z^{2}}{2}+h^{2} w z \\
&+\left(h^{2}-h +1\right) z +\frac{h^{2} w^{2}}{2}+h \left(h -2\right) w +\frac{h \left(h -2\right)}{2},
\\
\Phi_4=&-\frac{h \left(h +4\right) x^{2}}{2}+h \left(h +2\right) x z +h \left(h +2\right) x w +h \left(h +2\right) x +h y -\frac{h^{2} z^{2}}{2}-h^{2} w z -h^{2} z\\
& -\frac{h^{2} w^{2}}{2}+\left(-h^{2}+h +1\right) w -\frac{h^{2}}{2}.
\end{align*}

An interesting issue is that the vector field $X$ commute with the constant vector field $Y$ associated to the Hamiltonian function $H_{y}$, whose differential system is $\{
\dot{x}=0,\,\dot{y}=-1,\,\dot{z}=-1,\,\dot{w}=1
\}$, that is: $[X,Y]=0$. This fact has led us to verify, using a computer algebra software, that this happens for each of the families of vector fields in Theorem \ref{t:Hamiltonians de R4}, as reads its statement (d). Furthermore, in the above example, the maps $\Phi_h$  commute with the KHK maps associated with $Y$, given by $\Psi_k(x,y,z,w)=(x,y-k,z-k,w+k)$. That is,  $\Phi_h\circ \Psi_k= \Psi_k\circ\Phi_h$ for all $h$ and $k\in\mathbb{R}$.

Finally, from Theorem \ref{t:Hamiltonians de R4}(b), the vector field $X$ is a Lie symmetry of $\Phi_{h}$, since an easy computation shows that $X_{|\Phi_{h}}=D\Phi_{h}X$. Interestingly, a computation also shows that $Y_{|\Phi_{h}}=D\Phi_{h}Y$ so the vector field $Y$ is also a Lie symmetry of $\Phi_{h}$.

As a final remark, we notice that, of course, using the extra first integral $H_{y}$, we can do a dimension reduction, so that the vector field $X$ is reduced to linear vector field in $\mathbb{R}^3$. This reduced vector field has a linear first integral that allows a new dimension reduction.

\medskip

\noindent \textbf{Example C.} We consider a particular case of the Hamiltonian of type $H_1$, given by
$$
H(x,y,z,w)=-2 x^{3}-4 x^{2} z -4 x^{2} w +\frac{1}{2} x z^{2}+x z w +\frac{1}{2} x w^{2}+x^{2}+3 x z +3 x w +\frac{1}{2} z^{2}+z w +\frac{1}{2} w^{2}.
$$
The associated vector field $X$ is given by
\begin{align*}
\dot{x}&=0,\\
\dot{y}&=-6 x^{2}-8 x z -8 x w +\frac{1}{2} z^{2}+z w +\frac{1}{2} w^{2}+2 x +3 z +3 w,\\
\dot{z}&=4 x^{2}-x z -x w -3 x -z -w,\\
\dot{w}&=-4 x^{2}+x z +x w +3 x +z +w,
\end{align*}
and it has an associated KHK map is $\Phi_{h}(x,y,z,w)=(\Phi_1,\Phi_2,\Phi_3,\Phi_4)$, with
\begin{align*}
\Phi_1=&x,\\
\Phi_2=&-6 h \,x^{2}-8 h x z -8 h x w +\frac{1}{2} h \,z^{2}+h z w +\frac{1}{2} h \,w^{2}+2 h x +y +3 h z +3 h w,
\\
\Phi_3=&4 h \,x^{2}-h x z -h x w -3 h x +(1-h) z -h w,
\\
\Phi_4=&-4 h \,x^{2}+h x z +h x w +3 h x +h z +h w +w.
\end{align*}
The vector field $X$ has another first integral
$$
H_{x}(x,y,z,w)=-6 x^{2}-8 x z -8 x w +\frac{1}{2} z^{2}+z w +\frac{1}{2} w^{2}+2 x +3 z +3 w,
$$
which is functionally independent on $H$. The linear Hamiltonian vector field $Y$ associated with $H_{x}$ is given by
\begin{align*}
\dot{x}&=0,\\
\dot{y}&=-12 x -8 z -8 w +2,\\
\dot{z}&=8 x -z -w -3,\\
\dot{w}&=-8 x +z +w +3,
\end{align*}
and it has the associated KHK map
\begin{align*}
\Psi_k(x,y,z,w)=&\left(x,y+k \left(-12 x -8 z -8 w +2\right), z+\left(8 x -z -w -3\right) k ,\right.\\
&\left.w+\left(-8 x +z +w +3\right) k\right).
\end{align*}
As in the Example B, a straightforward computation shows that $[X,Y]=0$ and, therefore, the vector fields commute;  also for all $k,h\in\mathbb{R}$, it holds that $\Phi_h\circ \Psi_k= \Psi_k\circ\Phi_h$.

Of course, since one of the components of the vector field $X$ is null and there exists another first integral which is functionally independent on $H$, the vector field $X$ admits suitable dimension reductions.

\section{Three-degrees of freedom Hamiltonian systems}\label{s:r6}

We consider the planar three-degrees of freedom Hamiltonian vector field
$$
X=-\frac{\partial H}{\partial y} \frac{\partial}{\partial x}+\frac{\partial H}{\partial x}\frac{\partial}{\partial y}
-\frac{\partial H}{\partial w} \frac{\partial}{\partial z}+\frac{\partial H}{\partial z}\frac{\partial}{\partial w}-\frac{\partial H}{\partial s} \frac{\partial}{\partial r}+\frac{\partial H}{\partial r}\frac{\partial}{\partial s},
$$
with cubic Hamiltonian function

 \begin{equation}\label{e:genham6}
H(x,y,z,w,r,s)=\displaystyle{
\sum\limits_{0< i+j+k+l+m+n\leq 3}} a_{ijklmn}x^iy^jz^kw^lr^ms^n.
 \end{equation}
 The associated two degrees of freedom Hamiltonian system is
\begin{equation}\label{e:ham6}
 \dot{x}=-H_y, \quad
\dot{y}=H_x,\quad
\dot{z}=-H_w, \quad
\dot{w}=H_z\, \quad
\dot{r}=-H_s, \quad
\dot{s}=H_r,
\end{equation}
which also can be re-written as $q_i=\partial H/\partial p_i$, $p_i=-\partial H/\partial q_i$, with $q_1=x,q_2=z,q_3=r$ and $p_1=y,p_2=w,p_3=s$.
The relationship between the fields in both notations and their associate KHK maps is explained in Lemma \ref{l:noulemanou} of Appendix \ref{s:appendixnou}.

By imposing condition \eqref{e:restriction} on a  Hamiltonian function \eqref{e:genham6}
with the aid of a computer algebra system, we obtain that the KHK map \eqref{e:phi} associated to a Hamiltonian system \eqref{e:ham6} preserves the original Hamiltonian function \eqref{e:genham6} if and only if
\begin{equation}\label{e:lambdesr6}
\nabla H({\bf{x}})^T \left(I-\frac{1}{2}h{\mathrm D}X({\bf{x}})\right)^{-1}X({\bf{x}})=\frac{h(\Lambda_0+\Lambda_2 h^2+\Lambda_4 h^4)}{\left|I-\frac{1}{2}h {\mathrm D}X({\bf{x}})\right|}=0,
\end{equation}
where $\Lambda_0,\Lambda_2,\Lambda_4$ are polynomials in the partial derivatives of first and second order of $H$ with respect $x,y,z,w,r,$ and $s$,
of degrees $3$, $5$ and $7$, respectively, in these partial derivatives (see Remark \ref{r:referee2b}, below). For instance, the polynomial $\Lambda_0$ is given by
\begin{align}\label{e:h0r6gen}
\Lambda_0= & 32\left( -  H_{x}^{2} H_{yy} + 32 H_{x} H_{y} H_{xy} - 32 H_{x} H_{z} H_{yw} + 32 H_{x} H_{w} H_{yz} + 32 H_{x} H_{s} H_{yr}\right. \nonumber\\
&- 32 H_{x} H_{r} H_{ys} -  H_{y}^{2} H_{xx} + 32 H_{y} H_{z} H_{xw} - 32 H_{y} H_{w} H_{xz} - 32 H_{y} H_{s} H_{xr} \\
&+ 32 H_{y} H_{r} H_{xs} -  H_{z}^{2} H_{ww} + 32 H_{z} H_{w} H_{zw} + 32 H_{z} H_{s} H_{wr} - 32 H_{z} H_{r} H_{ws}  \nonumber\\
&\left.-  H_{w}^{2} H_{zz} - 32 H_{w} H_{s} H_{zr}+ 32 H_{w} H_{r} H_{zs} -  H_{s}^{2} H_{rr} + 32 H_{s} H_{r} H_{rs} -  H_{r}^{2} H_{ss}\right).\nonumber
\end{align}
We do not reproduce $\Lambda_2$ and $\Lambda_4$ here because they have $450$ and $2073$ coefficients, respectively,  as polynomials in the partial derivatives.

The polynomials $\Lambda_0,\Lambda_2,\Lambda_4$ have degree $5,7$ and $9$ in the variables $x,y,z,w,r$ and $s$, respectively, so they have $462, 1716$ and $5005$ coefficients respectively.
Equaling them to zero we obtain that the Hamiltonians in $\R^6$ satisfying condition \eqref{e:restriction}, must satisfy a system of $7183$ equations in their coefficients.  We have tried to solve the system using Maple in a parallel computer server formed by 9 nodes with 512GB RAM memory, in a Beowulf configuration, but we have not succeeded.

By imposing some additional restrictions on the equations  we have been able to find several dozens of non-trivial solutions. However, we are far from being able to claim that we have achieved all the solutions. In any case, as we will see below, the most interesting thing is that we have obtained cases for which the associated vector field \eqref{e:ham6} is not a Lie symmetry of the associated KHK map:

\begin{proposition}\label{p:propor6}
There exist Hamiltonian vector fields of $\R^6$ with Hamiltonian functions \eqref{e:genham6} of degree at most three, for which their associated KHK maps preserve the original Hamiltonian function. Furthermore, not all of them give rise to Hamiltonian vector fields of the form \eqref{e:ham6} which are Lie symmetries of the corresponding KHK maps.
\end{proposition}

\begin{proof}
Imposing, for example, $H_x=H_y$ (the choice of the pair $x$ and $y$ is not relevant, analogous results would be obtained by imposing $H_z=H_w$ or $H_r=H_s$) and solving with Maple the equation $ \Lambda_0=0$  we obtain a list of non-trivial solutions which also satisfy Equation \eqref{e:lambdesr6}. Among them, the one that corresponds to the $10$-parameter family of Hamiltonian systems given in Proposition \ref{p:propoappendixc}  n Appendix \ref{s:appendixc}. A simple example in this family is:

\begin{equation}\label{e:hamiltoniamesnou}
H=\left(x+y\right)\left(azr+w\right)+bzs.
\end{equation}
with associated vector field 
$$X=\left(-azr-w,azr+w,-x-y,ar(x+y)+bs,-bz,az(x+y)\right),$$
whose KHK map is $\Phi_h=(\Phi _{1,h},\Phi _{2,h},\Phi _{3,h},\Phi _{4,h},\Phi _{5,h},\Phi _{6,h})$ with
\begin{align}
\Phi_{1,h}=&\left( -h a z \left(h^{2} b (x +y) -2 h b z +4 r \right) +4 x -4  h w -2  b \,h^{2} s\right)/4,\nonumber\\
\Phi_{2,h}=&\left(h a z \left(h^{2} b (x+y) -2 h b z +4 r \right)+4 y+4 h w +2 b \,h^{2} s\right)/4,\nonumber\\
\Phi_{3,h}=&-h (x+y)+z,\label{e:novaphir6}\\
\Phi_{4,h}=&h a (x+y) r +w +b h s,\nonumber\\
\Phi_{5,h}=&\left(h^{2} b (x +y) -2 h b z +2 r\right)/2,\nonumber\\
\Phi_{6,h}=&\left( -a h \left(x +y \right) \left(h (x +y) -2 z \right)+2s   \right)/2.\nonumber
\end{align}
A straightforward computation shows that 
$$X{|\Phi_h}-{\mathrm D}\Phi_h\,X=
\left( \frac{ ab \left(x +y \right)^{2}h^{3} }{4}, -\frac{ ab \left(x +y \right)^{2}h^{3} }{4} , 0 , 0 , 0 , 0 \right)^t.
$$ Hence, if  $ab\neq 0$, then the  vector field $X$ is not a Lie symmetry of $\Phi_h$.

\end{proof}

The following observations are due to a reviewer of a previous version. We have not delved into these questions in this work.

\begin{remark}\label{r:referee2a}
The direct inspection of the components of the first iterates of the KHK map $\Phi_h$ in Equation \eqref{e:novaphir6}, indicates that the (algebraic) degree of the first iterates remains constant, which, if confirmed, would imply that the algebraic entropy of this map is $0$, as well as the possible existence of a linearization and additional first integrals, \cite{BV,V}. 
\end{remark}

\begin{remark}\label{r:referee2b} Based on the formulas in the proof of Lemma \ref{l:invr2}  as well as Equations \eqref{e:h0h2r4} and \eqref{e:lambdesr6}, it appears that a pattern may emerge regarding the conditions for exact preservation of the Hamiltonian. It seems that
for a general system with $d$-degrees of freedom, the condition \eqref{e:restriction} can be expressed as:  
$$
\nabla H({\bf{x}})^T \left(I-\frac{1}{2}h{\mathrm D}X({\bf{x}})\right)^{-1}X({\bf{x}})=\frac{h}{\left|I-\frac{1}{2}h {\mathrm D}X({\bf{x}})\right|}\sum_{\ell=0}^{d-1} h^{2\ell}\Lambda_{2\ell}=0
$$
for some suitable polynomials $\Lambda_{2\ell}$ of degree $2(\ell + 1) + 1$ in the derivatives of $H$.  
\end{remark}

\section{Some considerations on the symplecticity of the maps}\label{s:add}

As we have indicated in the introduction, our initial purpose was to examine analytically whether the Hamiltonian fields $X_H$ in $\mathbb{R}^n$ for which their associated KHK maps preserve the same cubic Hamiltonian $H$ were also Lie symmetries of these maps. From Proposition \ref{p:propor6}, we now know that this is not true in general, although it is true in all the cases found in $\mathbb{R}^2$ and $\mathbb{R}^4$.

Once this point was clarified,  a second objective was to explore if under the initial hypothesis of the preservation of the original Hamiltonian $H$, the fact that the field $X_H$ is a Lie symmetry of the
KHK map $\Phi_h$ is related with the symplecticity of the map. In this section, we will see that if a symplectic map $\Phi$ (not  necessarily KHK) preserves a $\mathcal{C}^1$ first integral,  then the Hamiltonian vector field $X_H$ is a Lie symmetry of $\Phi$. The converse is true for planar maps (Proposition \ref{p:corolpla}).

Set $\mathbf{x}\in\mathbb{R}^{2n}$. A map $\Phi(\mathbf{x})$ is \emph{symplectic} if and only if
$\mathrm{D}\Phi(\mathbf{x})^t\,\Omega\,\mathrm{D}\Phi(\mathbf{x})=\Omega$, where
$$
\Omega=\left(\begin{array}{cc}
\mathbf{0} & {I_n}\\
-{I_n} & \mathbf{0}
\end{array}\right),
$$ see \cite{Marsden1999}. By using that a matrix $M$ is symplectic if and only if $M^t$ is symplectic, we will use the following equivalent condition of symplecticity:
\begin{equation}\label{e:eqsympl}
\mathrm{D}\Phi(\mathbf{x})\,\Omega\,\mathrm{D}\Phi(\mathbf{x})^t=\Omega.
\end{equation}

For general $\mathcal{C}^1$ maps, we have the following result:
\begin{proposition}\label{p:symplsL}
Let $\Phi$ be a symplectic $\mathcal{C}^1$ map defined in an open set $\mathcal{U}\subseteq \mathbb{R}^{2n}$. Let $H$ be a $\mathcal{C}^1$ first integral of $\Phi$ in $\mathcal{U}$. Then, the Hamiltonian vector field $X_H=\Omega\nabla H$ is a Lie symmetry of $\Phi$.
\end{proposition}
\begin{proof}
Observe that if $H$ is a first integral of $\Phi$, then $\nabla H(\Phi(\mathbf{x}))^t\mathrm{D}\Phi(\mathbf{x})=\nabla H(\mathbf{x})^t$, hence
\begin{equation}\label{e:nabla}
\nabla H(\mathbf{x})=\mathrm{D}\Phi(\mathbf{x})^{t}\, \nabla H(\Phi(\mathbf{x})).
\end{equation}
We have to prove that the compatibility relation \eqref{e:Lie-sym-char} holds. By using the symplecticity condition  \eqref{e:eqsympl} and Equation \eqref{e:nabla} we have,
$$
X_{H|\Phi({\bf{x}})}=\Omega\, \nabla H(\Phi(\mathbf{x}))={\mathrm D}\Phi({\bf{x}})\,\Omega\,{\mathrm D}\Phi({\bf{x}})^t\, \nabla H(\Phi(\mathbf{x}))= {\mathrm D}\Phi({\bf{x}})\,\Omega\,\nabla H(\mathbf{x})={\mathrm D}\Phi({\bf{x}})\,X_H({\bf{x}}).
$$
\end{proof}

Proposition \ref{p:symplsL} agrees with the result of Ge  (that can be found in the work of Ge and Marsden, \cite[p. 135]{GeMarsden}), that states that if a discretization method for a Hamiltonian vector field \emph{with no other first integral}, preserves the Hamiltonian function and it is symplectic, then it is the \emph{time advance map} up to a reprametrization of time. This \emph{time advance property} happens when a field $X$ is a Lie symmetry of a $\mathcal{C}^1$ map $\Phi$ \emph{which preserves the orbits of $X$}, see  \cite[Proposition 6]{CGM08}, and also Section \ref{s:defi}.

The converse of Proposition \ref{p:symplsL} is true for fields in $\mathbb{R}^2$:

\begin{proposition}\label{p:corolpla}
Let $\Phi$ be a $\mathcal{C}^1$ map defined in an open set $\mathcal{U}\subseteq \mathbb{R}^2$, and let $H$ be a $\mathcal{C}^1$ first integral of $\Phi$ in $\mathcal{U}$. Then, the \emph{planar} Hamiltonian vector field $X_H=\Omega \nabla H$ is a Lie symmetry of $\Phi$ if and only if the map is symplectic.
\end{proposition}

\begin{proof}
Taking into account Proposition \ref{p:symplsL}, we only have to prove the converse relation, that is, that if $X_H$ is a Lie symmetry  of the map $\Phi$, then the map is symplectic. A straightforward computation shows that for $\mathbf{x}\in\mathbb{R}^2$,
$$
{\mathrm D}\Phi({\bf{x}})\,\Omega\,{\mathrm D}\Phi({\bf{x}})^t\,-\Omega=(|\mathrm{D}\Phi({\bf{x}})|-1)\,\Omega.
$$
Hence, a planar $\mathcal{C}^1$ map $\Phi$ defined in an open set $\mathcal{U}\subseteq \mathbb{R}^2$ is symplectic if and only if $|\mathrm{D}\Phi({\bf{x}})|=1$ for $\mathbf{x}\in\mathcal{U}$. On the other hand, as mentioned in Section \ref{s:defi}, a vector field
 $$X(\mathbf{x})=\frac{1}{\nu(\mathbf{x})}\left(-H_y,H_x\right)=-\frac{1}{\nu(\mathbf{x})}\Omega \nabla H,$$ is a Lie symmetry of a map $\Phi$, with first integral $H$, if and only if $\nu$ is
an invariant measure.  As a consequence, the  Hamiltonian field  $X=(-H_y,H_x)=-\Omega\nabla H=-X_H$ is a Lie symmetry of a map $\Phi$ in $\mathcal{U}$ if and only if the map preserves the measure with density $\nu(\mathbf{x})= 1$ for $\mathbf{x}\in\mathcal{U}$, but in this case, from Equation \eqref{e:mesura}, we get that $|\mathrm{D}\Phi({\bf{x}})|= 1$. Of course, if $X$ is a Lie symmetry so is $X_H=-X$, and the result follows.
\end{proof}

\begin{corollary}\label{c:simplecticr2} The planar KHK maps associated with the Hamiltonian fields in Theorem~\ref{t:Hamiltonians de R2}  are symplectic.
\end{corollary}

\medskip

Note that, the Hamiltonian vector fields in Sections \ref{s:r4} and \ref{s:r6} have the form $X(\mathbf{x})=B\nabla H(\mathbf{x})$, where $B$ is the $(2n)\times(2n)$ matrix given  by the anti-symmetric matrix
\begin{equation}\label{e:B}
B=\left(\begin{array}{crccr}
0 & -1 & & 0 & 0\\
1 & 0 & &0 & 0\\
& & \ddots & &\\
0 & 0 & &0 & -1\\
0 & 0 & & 1 & 0
\end{array}\right).
\end{equation}

The Lemma \ref{l:noulemanou} in Appendix \ref{s:appendixnou} relates this kind of Hamiltonian vector fields and their associated KHK maps with the Hamiltonian fields written in canonical form and their associated KHK map. In summary, a field of the form $X(\mathbf{x})=B\nabla H(\mathbf{x})$ with associated KHK maps $\Phi_h(\mathbf{x})$, where $\mathbf{x}=(x_1,\ldots,x_{2n})\in\mathbb{R}^{2n}$,
conjugates through the change of variables
$x_1=q_1$, $x_2=p_1$, $x_3=q_2$, $x_4=p_2,\ldots,x_{2n-1}=q_n$, $x_{2n}=p_n$, and a time parametrization $t\to -t$, with the field  $$Y_{\tilde{H}}(\mathbf{y})=\Omega\,\nabla \tilde{H}(\mathbf{y}),$$ where $\mathbf{y}=(q_1,\ldots,q_{n},$ $p_1,\ldots,p_{n})\in\mathbb{R}^{2n}$, and
$\tilde{H}(\mathbf{y})=H(P\mathbf{y})$, being $P$ the permutation matrix such that $\mathbf{x}=P\mathbf{y}$. Furthermore, the KHK map associated with $Y_{\tilde{H}}(\mathbf{y})$ is given by $$\Psi_h(\mathbf{y})=P^t\Phi_{-h}(P\mathbf{y}).$$

Using Lemma \ref{l:noulemanou} and by a verification carried out by using Maple, we obtain the statement of the following result concerning the KHK maps associated with the vector fields  in Theorem~\ref{t:Hamiltonians de R4} for vector fields in $\mathbb{R}^4$.

\begin{proposition}\label{p:totesymplectiquesr4}
Let $\Phi_h=\left(\Phi _{1,h},\Phi _{2,h},\Phi _{3,h},\Phi _{4,h}\right)$ be any of the KHK maps associated with the Hamiltonian fields $X$ in Theorem~\ref{t:Hamiltonians de R4}. Then, the map $${\Psi}_{h}(q_1,q_2,p_1,p_2)=\left(\Phi _{1,-h},\Phi _{3,-h},\Phi _{2,-h},\Phi _{4,-h}\right)(q_1,p_1,q_2,p_2)$$ is symplectic.
\end{proposition}
Notice that, from Lemma \ref{l:noulemanou}(c), to prove the above result it is only necessary to check that the KHK maps $\Phi_{h}$ associate with the Hamiltonian vector field in Theorem \ref{t:Hamiltonians de R4} satisfy
\begin{equation}\label{e:equacioambB}
{\mathrm D}\Phi_h({\bf{x}})\,B\,{\mathrm D}\Phi_h({\bf{x}})^t=B.
\end{equation}

\smallskip

As a consequence of  Propositions \ref{p:propor6} and \ref{p:symplsL}, and Lemma \ref{l:noulemanou}(d), for vector fields in $\mathbb{R}^6$, we obtain
\begin{corollary}\label{c:corolarinosimplr6}
The map $${\Psi}_h(q_1,q_2,q_3,p_1,p_2,p_3)=\left(\Phi _{1,-h},\Phi _{3,-h},\Phi _{5,-h},\Phi _{2,-h},\Phi _{4,-h},\Phi _{6,-h}\right)(q_1,p_1,q_2,p_2,q_3,p_3),$$ whe\-re $\Phi_h=\left(\Phi _{1,h},\Phi _{2,h},\Phi _{3,h},\Phi _{4,h},\Phi _{5,h},\Phi _{6,h}\right)$ is any of the  KHK maps associated with the Hamiltonians of $\mathbb{R}^6$ for which the field $X$ is not a Lie symmetry of $\Phi_h$, is not symplectic.
\end{corollary}

The following observation is due to a reviewer of an earlier version of this paper:

\begin{remark}\label{r:referee2}  
The preceding result demonstrates that the KHK map $\Phi_h$ in the counterexample given in the proof of Proposition \ref{p:propor6}, with associate Hamiltonian \eqref{e:hamiltoniamesnou}, is not symplectic. However, it 
is worth noticing that it still preserves the deformation of the symplectic structure \eqref{e:equacioambB} given by
$ B_h= B + h^2 B_2,$  
where  
$$ 
B_2 = \left(\begin{array}{cccccc}
0 &0&0 &-\alpha &{b}/{4} & 0\\
0 &0&0 &\alpha &-{b}/{4} & 0\\
0 &0&0 & 0 & 0 & 0\\
\alpha &-\alpha&0 & 0 & 0 & 0\\
-b/4 &b/4&0 & 0 & 0 & 0\\
0 &0&0 & 0 & 0 & 0
\end{array} \right),
$$being $\alpha$ any arbitrary constant.

\end{remark}

\section{Conclusions}

We have investigated the set of cubic Hamiltonian vector fields for which their associated Kahan-Hirota-Kimura  maps preserve the original Hamiltonian function. For fields in $\mathbb{R}^2$ and $\mathbb{R}^4$, we have identified nontrivial examples and in all the cases we found, the associated fields are Lie symmetries of the corresponding 
KHK maps, which preserve a symplectic structure. In the planar case, all the fields we identified correspond to factorizable Hamiltonians, representing singular cases of those characterized in \cite{GMQ24}, with associated trivial dynamics.  
In contrast, for $\mathbb{R}^6$, we have discovered nontrivial cases where the field is not a Lie symmetry of the KHK map.  

Our study has been limited to even-dimensional cases where a Poisson structure is present. Additional work could potentially be conducted in odd dimensions,  particularly in the three-dimensional case, where the KHK method and its extensions have proven to be fruitful \cite{ASW,GG23,PS10}.  

Similarly, it is natural to explore what happens to the KHK maps when considering vector fields with an integrating factor, a Jacobi multiplier \cite{BeGia}, or those of the form $X(\mathbf{x}) = B(\mathbf{x})\nabla H(\mathbf{x}) $, that is, when the anti-symmetric matrix $ B(\mathbf{x}) $ is no longer constant. For some of these systems, many interesting results have been obtained, which could be analyzed within the framework of exact preservation; see, for instance, \cite{Cel17,GJ20,Kam19,PZ17}.

\bigskip

 \centerline{\textbf{Acknowledgments}}

\medskip

We would like to express our sincere gratitude to the anonymous reviewers of an earlier version of this work for their careful reading and for the insightful and generous suggestions, which are reflected in Remarks \ref{r:remark7}, \ref{r:referee2a}, \ref{r:referee2b} and \ref{r:referee2}, as well as in the conclusions section. We also would like to thank Dr.~Benjamin K. Tapley for his kind comments and references that helped us to improve a preliminary version of the paper.
\medskip

The authors are supported by the Ministry of Science and Innovation--State Research Agency of the
Spanish Government through grant PID2022-136613NB-I00.
 They also acknowledge the 2021 SGR 01039 consolidated research groups recognition from Ag\`{e}ncia de Gesti\'{o} d'Ajuts Universitaris i de Recerca, Generalitat de Catalunya.

\vfill
\newpage
\appendix

\section{Appendix}\label{s:appendixnou}

We consider the real vectors
$\mathbf{x}=(x_1,\ldots,x_{2n})$ and $\mathbf{y}=(q_1,\ldots,q_{n},$ $p_1,\ldots,p_{n})$ related through the change of variables
$x_1=q_1$, $x_2=p_1$, $x_3=q_2$, $x_4=p_2,\ldots,x_{2n-1}=q_n$, $x_{2n}=p_n$ or, in other words,   via $\mathbf{x}=P\mathbf{y}$ where $P$ is a permutation matrix.
The permutation matrices in $\mathbb{R}^4$ and $\mathbb{R}^6$ are given by
$$
P=\left(\begin{array}{cccc}
1&0&0&0\\
0&0&1&0\\
0&1&0&0\\
0&0&0&1
\end{array}\right)\mbox{ and } P=\left(\begin{array}{cccccc}
1&0&0&0&0&0\\
0&0&0&1&0&0\\
0&1&0&0&0&0\\
0&0&0&0&1&0\\
0&0&1&0&0&0\\
0&0&0&0&0&1
\end{array}\right)
$$
respectively.

The following result relates the fields of the form $X(\mathbf{x})=B\nabla H(\mathbf{x})$, where $B$ is given by expression \eqref{e:B}, and their associated KHK maps with the Hamiltonian fields written in canonical form and their associated KHK maps.

\begin{lemma}\label{l:noulemanou}
Let $H$ be a $\mathcal{C}^1$ function in an open set of $\mathbb{R}^{2n}$. The following statements hold:
\begin{enumerate}
\item[(a)] Through the change of variables $\mathbf{x}=P\mathbf{y}$ and the time reparametrization $t\to-t$, the Hamiltonian vector field $X(\mathbf{x})=B\,\nabla H(\mathbf{x})$ is conjugate with
$Y_{\tilde{H}}(\mathbf{y})=\Omega\,\nabla \tilde{H}(\mathbf{y})$, with $\tilde{H}(\mathbf{y})=H(P\mathbf{y})$.
\item[(b)] Let $\Psi_h$ and $\phi_h$ the KHK maps associated with the vector fields $Y_{\tilde{H}}$ and $X_{H}$, respectively, then
$$
\Psi_h(\mathbf{y})=P^t\Phi_{-h}(P\mathbf{y}).
$$
\item[(c)] The map $\Psi_h$ is symplectic if and only if
$$
{\mathrm D}\Phi_h({\bf{x}})\,B\,{\mathrm D}\Phi_h({\bf{x}})^t=B.
$$
\item[(d)] The vector field $X$ is a Lie symmetry of $\Phi_h$ if and only if the vector field $Y_{\tilde{H}}$ is a Lie symmetry of $\Psi_h$.
\end{enumerate}
\end{lemma}

To prove the above result, we need the following technical Lemma, that uses the matrix identity
\begin{equation}\label{e:identitat}
(I+UV)^{-1}U=U(I+VU)^{-1},
\end{equation}
where $U$ and $V$ are conformable matrices, which is obtained directly from the identity $U(I+VU)=(I+UV)U$. We also recall that since $P$ is a permutation matrix, then $P^{-1}=P^t$.

\begin{lemma}\label{l:noulema}
Let $\Phi_h(\mathbf{x})$ be the KHK map associated with a vector field $X(\mathbf{x})$. Let $\mathbf{x}=P\mathbf{y}$  and $\Psi_h(\mathbf{y})$ be the KHK map associated with the vector field $Y(\mathbf{y})=P^t\,X(P\mathbf{y})$. Then
$$
\Psi_h(\mathbf{y})=P^t\Phi_h(P\mathbf{y}).
$$
\end{lemma}

\begin{proof}
Setting $U=P^t$, $V=-\frac{1}{2}h{\mathrm D}X({P\bf{y}})P$, using the identity \eqref{e:identitat} and $P^{-1}=P^t$, we have that the KHK map associated with $Y(\mathbf{y})$ satisfies

\begin{align*}
\Psi_h(\mathbf{y})&=\mathbf{y}+h\left(I-\frac{1}{2}h {\mathrm D}Y({\bf{y}})\right)^{-1} Y({\bf{y}})\\
&= \mathbf{y}+h\left(I-\frac{1}{2}h P^t{\mathrm D}X({P\bf{y}})P\right)^{-1} P^tX({P\bf{y}})\\
&= P^t\mathbf{x}+hP^t\left(I-\frac{1}{2}h {\mathrm D}X({P\bf{y}})PP^t\right)^{-1} X({P\bf{y}})\\
\end{align*}

\begin{align*}
&= P^t\left(\mathbf{x}+h\left(I-\frac{1}{2}h {\mathrm D}X({\mathbf{x}})\right)^{-1} X({\mathbf{x}})\right)\\
&=P^t\phi_h(\mathbf{x})=P^t\phi_h(P\mathbf{y}).
\end{align*}
\end{proof}

\begin{lemma}\label{l:ultim}
Set $\mathbf{x},\mathbf{y}\in \mathbb{R}^m$,
and let $\mathbf{x}=P\mathbf{y}$ where $P$ is a permutation matrix. Then, the vector field $X(\mathbf{x})$ is a Lie symmetry of a $\mathcal{C}^1$ map $\Phi(\mathbf{x})$ if and only if the vector field $Y(\mathbf{y})=P^t\,X(P\mathbf{y})$ is a Lie symmetry of the map $\Psi(\mathbf{y})=P^t\Phi_h(P\mathbf{y})$.
\end{lemma}
\begin{proof}
Using the compatibility equation \eqref{e:Lie-sym-char},
$
X_{|\Phi(\mathbf{x})}=\mathrm{D}\Phi(\mathbf{x})\,X(\mathbf{x}),
$ and the fact that $P^{-1}=P^t$ we obtain
$$
X_{|\Phi(P\mathbf{y})}=\mathrm{D}\Phi(P\mathbf{y})\,X(P\mathbf{y})
\,\Leftrightarrow\,
X_{|\Phi(P\mathbf{y})}=P\,\mathrm{D}\Psi(\mathbf{y})P^t\,X(P\mathbf{y})\,\Leftrightarrow
$$
$$
X_{|P\Psi(\mathbf{y})}=P\,\mathrm{D}\Psi(\mathbf{y})P^t\,X(P\mathbf{y})
\,\Leftrightarrow\,
P^t\,X_{|P\Psi(\mathbf{y})}=\mathrm{D}\Psi(\mathbf{y})P^t\,P\,Y(\mathbf{y})\Leftrightarrow
$$
$$
Y_{|\Psi(\mathbf{y})}=\mathrm{D}\Psi(\mathbf{y})\,Y(\mathbf{y}).
$$
\end{proof}

\begin{proof}[Proof of Lemma \ref{l:noulemanou}]
(a) If $\tilde{H}(\mathbf{y})=H(P\mathbf{y})$, then $\nabla \tilde{H}(\mathbf{y})=P^t\nabla H(P\mathbf{y})$.  A computation shows that
\begin{equation}\label{e:menysomega}
P^t\,B\,P=-\Omega.
\end{equation} Using this equation and
 that $P^{-1}=P^t$, we have
$$
\dot{\mathbf{y}}=P^t\dot{\mathbf{x}}=P^t\,B\,\nabla H(\mathbf{x})=P^t\,B\,\nabla H(P\mathbf{y})=
P^t\,B\,P\, \nabla \tilde{H}(\mathbf{y})=-\Omega \nabla \tilde{H}(\mathbf{y}).
$$ By using the time reparametrization $t\to -t$ we obtain the result.

Statement (b) is a direct consequence of the Lemma \ref{l:noulema} and the well know fact that  time reparametrization $t\to -t$ affects KHK maps only in the fact that the discretization step must be changed so that $h\to -h$.

(c) Assume that $\Psi_h$ is symplectic. From equation
$${\mathrm D}\Psi_h({\bf{y}})\,\Omega\,{\mathrm D}\Psi_h({\bf{y}})^t=\Omega,$$
and by using that $\Psi_h(\mathbf{y})=P^t\Phi_{-h}(P\mathbf{y})$, we get
$$
P^t{\mathrm D}\Phi_{-h}({\bf{x}})\,P\,\Omega\,P^t\,{\mathrm D}\Phi_{-h}({\bf{x}})^tP=-P^t\,B\,P.
$$
By Equation \eqref{e:menysomega}, we have $P\Omega P^t=-B$, so we finally obtain
$$
{\mathrm D}\Phi_h({\bf{x}})\,B\,{\mathrm D}\Phi_h({\bf{x}})^t=B.
$$The converse statement is obtained, by reversing the computations.

Statement (d) is a direct consequence of Lemma \ref{l:ultim}.
\end{proof}

\newpage

\section{Appendix}\label{s:appendixa}

  Here, we provide the list of
$54$ Hamiltonian functions \eqref{e:genham4} of degree at most three, referenced in Theorem \ref{t:Hamiltonians de R4},
giving rise to Hamiltonian vector fields for which their associated KHK maps preserve the original Hamiltonian.

{\small
\begin{flalign*}
H_{1} &=
a_{1000} x +a_{2000} x^{2}+a_{3000} x^{3}+\frac{a_{0001} a_{1011} z}{2 a_{1002}}+\frac{a_{1001} a_{1011} x z}{2 a_{1002}}&&\\
&+\frac{a_{1011} a_{2001} x^{2} z}{2 a_{1002}}
+\frac{a_{0002} a_{1011}^{2} z^{2}}{4 a_{1002}^{2}}+\frac{a_{1011}^{2} x \,z^{2}}{4 a_{1002}}+a_{0001} w +a_{1001} x w
&&\\
&+a_{2001} x^{2} w +\frac{a_{0002} a_{1011} z w}{a_{1002}}+a_{1011} x z w +a_{0002} w^{2}+a_{1002} x \,w^{2}.\\
\end{flalign*}
}

\vspace{-1.6cm}

{\small
\begin{flalign*}
H_{2} &=
a_{1000} x +a_{0001} w +a_{0010} z +a_{0100} y +\frac{a_{1001} a_{1011} x z}{2 a_{1002}}+\frac{a_{1011}^{2} x \,z^{2}}{4 a_{1002}}+a_{1001} x w&&\\
&
-\frac{a_{1011} \left(a_{0001} a_{1011}-2 a_{0010} a_{1002}\right) x^{2} z}{2 a_{0100} a_{1002}} -\frac{\left(a_{0001} a_{1011}-2 a_{0010} a_{1002}\right) x^{2} w}{a_{0100}}&&\\
&+\frac{a_{0002} a_{1011} z w}{a_{1002}}+a_{1002} x \,w^{2}+a_{1011} x z w
+a_{0002} w^{2}
&&\\
&+\frac{\left(a_{0001} a_{1011}
-2 a_{0010} a_{1002}\right)^{2} x^{3}}{4 a_{0100}^{2} a_{1002}}-\frac{x^{2}\left(a_{0001} a_{1011}-2 a_{0010} a_{1002}\right) A}{4 a_{0100}^{2} a_{1002}^{2}}&&\\
&
+\frac{a_{0002} a_{1011}^{2} z^{2}}{4 a_{1002}^{2}}.
\end{flalign*}
}
where $A=a_{0001} a_{0002} a_{1011}-2 a_{0002} a_{0010} a_{1002}+2 a_{0100} a_{1001} a_{1002}$.

\vspace{-0.2cm}

{\small
\begin{flalign*}
H_{3}&=
a_{1000} x +a_{0001} w +\frac{A_1 z}{2 a_{0101} a_{1002}}+a_{0100} y\frac{A_2x z}{2 a_{1002}^{2} a_{0101}}+a_{0002} w^{2}&&\\
 &-\frac{a_{1011} \left(a_{0002} a_{1011}-a_{0011} a_{1002}\right) x^{2} z}{a_{0101} a_{1002}}
 +\frac{a_{1011}^{2} x \,z^{2}}{4 a_{1002}}+a_{1001} x w+a_{0011} z w  &&\\
 &-\frac{2 \left(a_{0002} a_{1011}-a_{0011} a_{1002}\right) x^{2} w}{a_{0101}}+a_{1002} x \,w^{2}+a_{1011} x z w &&\\
 & -\frac{\left(a_{0002} a_{1011}-a_{0011} a_{1002}\right) x y}{a_{1002}}-\frac{\left(a_{0002} a_{1011}-a_{0011} a_{1002}\right) A_3 x^{2}}{a_{0101}^{2} a_{1002}^{2}}&&\\
 &+\frac{a_{1011} a_{0101} y z}{2 a_{1002}}+a_{0101} y w-\frac{a_{1011} \left(a_{0002} a_{1011}-2 a_{0011} a_{1002}\right) z^{2}}{4 a_{1002}^{2}}&&\\
 &+\frac{\left(a_{0002} a_{1011}-a_{0011} a_{1002}\right)^{2} x^{3}}{a_{0101}^{2} a_{1002}},
 \end{flalign*}
}where
{\small \begin{align*}
A_1&=a_{0001} a_{0101} a_{1011}-2 a_{0002} a_{0100} a_{1011}+2 a_{0011} a_{0100} a_{1002},\\
A_2&=2 a_{0002}^{2} a_{1011}^{2}-4 a_{0002} a_{0011} a_{1002} a_{1011}+2 a_{0011}^{2} a_{1002}^{2}\\
&+a_{0101} a_{1001} a_{1002} a_{1011},\\
A_3&=a_{0002}^{2} a_{1011}-a_{0002} a_{0011} a_{1002}+a_{0101} a_{1001} a_{1002}.
\end{align*}}

\vspace{-1cm}

{\small
\begin{flalign*}H_{4,5}& =
a_{0001} w +\frac{\left(a_{0100} a_{0110} a_{1001}^{2}\pm\sqrt{A}\right) z}{2 a_{0110} a_{1001} a_{0002}} +\frac{a_{0002} a_{0110} y^{2}}{a_{1001}}+a_{1010} x z&&\\
& +a_{1001} x w +\frac{2 a_{0002} a_{1010} x y}{a_{1001}}+a_{0110} y z +a_{0002} w^{2}\\
&+\frac{\left(4 a_{0002}^{2} a_{1010}^{2}+a_{0110} a_{1001}^{3}\right) x^{2}}{4 a_{0002} a_{0110} a_{1001}}+a_{1000} x +a_{0100} y+\frac{a_{1001} a_{0110} z^{2}}{4 a_{0002}},
 \end{flalign*}
}
where $$A=-a_{0110} a_{1001} \left(a_{0001} a_{0110} a_{1001}+2 a_{0002} a_{0100} a_{1010}-2 a_{0002} a_{0110} a_{1000}\right)^{2}.$$

\vspace{-0.4cm}

{\small
\begin{flalign*}
H_{6} &=
a_{2001} x^{2} w +a_{3000} x^{3}+a_{0002} w^{2}+a_{1001} x w +a_{2000} x^{2}+a_{0001} w +a_{1000} x.&&\\
\end{flalign*}
}

\vspace{-1.6cm}

 {\small
 \begin{flalign*}
 H_{7} &=
a_{1000} x -\frac{a_{0010} \left(a_{0002} a_{0010}-a_{0100} a_{1001}\right) x^{2}}{a_{0100}^{2}}+a_{0100} y +a_{0010} z&&\\
& +a_{0001} w +a_{1001} x w +a_{0002} w^{2}.
\end{flalign*}}

\vspace{-1cm}

{\small
\begin{flalign*}
H_{8}&=
\frac{a_{0100} a_{1100} x}{2 a_{0200}}+\frac{a_{1100}^{2} x^{2}}{4 a_{0200}}+a_{0100} y +a_{1100} x y +a_{0200} y^{2}+a_{0010} z &&\\
&+\frac{a_{0110} a_{1100} x z}{2 a_{0200}}+\frac{a_{1100}^{2} a_{0210} x^{2} z}{4 a_{0200}^{2}}+a_{0110} y z+\frac{a_{1100} a_{0210} x y z}{a_{0200}}+a_{0210} y^{2} z &&\\
&+a_{0020} z^{2}+\frac{a_{0120} a_{1100} x z^{2}}{2 a_{0200}}+a_{0120} y z^{2}+a_{0,0,3,0} z^{3}.
\end{flalign*}
}

\vspace{-1cm}

{\small
\begin{flalign*}H_{9} &=
\frac{a_{0100} a_{1110} x}{2 a_{0210}}+a_{0100} y +a_{0010} z +\frac{a_{0110} a_{1110} x z}{2 a_{0210}}+\frac{a_{1110}^{2} x^{2} z}{4 a_{0210}}+a_{0110} y z &&\\
&+a_{1110} x y z +a_{0210} y^{2} z +a_{0020} z^{2}+\frac{a_{0120} a_{1110} x z^{2}}{2 a_{0210}}+a_{0120} y z^{2}+a_{0,0,3,0} z^{3}.
\end{flalign*}
}

\vspace{-1cm}

{\small
\begin{flalign*}
H_{10} &=
a_{1020} x z^{2}+a_{0120} y z^{2}+a_{0,0,3,0} z^{3}+a_{1010} x z +a_{0110} y z +a_{0020} z^{2}+a_{1000} x &&\\& +a_{0100} y +a_{0010} z.&&\\
\end{flalign*}
}

\vspace{-1.6cm}

{\small
\begin{flalign*}
H_{11} &=
a_{3000} x^{3}+a_{2010} x^{2} z +a_{1020} x z^{2}+a_{0,0,3,0} z^{3}+a_{2000} x^{2}+a_{1010} x z +a_{0020} z^{2}&&\\&+a_{1000} x +a_{0010} z.
\end{flalign*}
}

\vspace{-1cm}

{\small
\begin{flalign*}H_{12} &=
\frac{a_{1110}^{2} a_{0200} x^{2}}{4 a_{0210}^{2}}+\frac{2 a_{0210} a_{1000} y}{a_{1110}}+\frac{a_{0200} a_{1110} x y}{a_{0210}}+a_{0200} y^{2}+a_{0010} z &&\\
&+a_{1010} x z +\frac{a_{1110}^{2} x^{2} z}{4 a_{0210}}+\frac{2 a_{0210} a_{1010} y z}{a_{1110}}+a_{1000} x +a_{1110} x y z +a_{0210} y^{2} z &&\\
&+\frac{a_{0120} a_{1110} x  z^{2}}{2 a_{0210}}+a_{0120} y  z^{2}+a_{0,0,3,0} z^{3}.
\end{flalign*}
}

\vspace{-1.1cm}

{\small
\begin{flalign*}
H_{13}& =+a_{1020} x z^{2}+
a_{1000} x +\frac{a_{1100}^{2} x^{2}}{4 a_{0200}}+\frac{2 a_{0200} a_{1000} y}{a_{1100}}+a_{1100} x y +a_{0200} y^{2}&&\\
&+a_{0010} z +a_{1010} x z +\frac{2 a_{0200} a_{1010} y z}{a_{1100}}+\frac{2 a_{0200} a_{1020} y z^{2}}{a_{1100}}+a_{0,0,3,0} z^{3}.
\end{flalign*}
}

\vspace{-0.8cm}

 {\small
\begin{flalign*}H_{14}&=
a_{1000} x +\frac{2 a_{0210} a_{1000} y}{a_{1110}}+a_{0010} z +a_{1010} x z +\frac{a_{1110}^{2} x^{2} z}{4 a_{0210}}+\frac{2 a_{0210} a_{1010} y z}{a_{1110}}&&\\
& +a_{1110} x y z+a_{0210} y^{2} z +a_{1020} x z^{2}+\frac{2 a_{0210} a_{1020} y z^{2}}{a_{1110}}+a_{0,0,3,0} z^{3}.
\end{flalign*}
}

\vspace{-1cm}

 {\small
\begin{flalign*}H_{15} &=
a_{2001} x^{2} w +a_{3000} x^{3}+a_{2010} x^{2} z +a_{1001} x w +a_{2000} x^{2}+a_{1010} x z +a_{0001} w  &&\\
&+a_{1000} x+a_{0010} z.
\end{flalign*}
}

\vspace{-1.2cm}{\small
\begin{flalign*}H_{16}&=
a_{1000} x +a_{1110} x y z +a_{0001} w +a_{0010} z +a_{0100} y +\frac{A^{2} z^{3}}{4 a_{0001}^{2} a_{2010}}+a_{1010} x z &&\\
&+\frac{a_{1010} a_{1110} y z}{2 a_{2010}}-\frac{a_{1010} A z^{2}}{2 a_{0001} a_{2010}}+a_{2010} x^{2} z +\frac{a_{1110}^{2} y^{2} z}{4 a_{2010}}-\frac{A x z^{2}}{a_{0001}}&&\\
&-\frac{a_{1110} A y z^{2}}{2 a_{0001} a_{2010}},
\end{flalign*}
}where $A=2 a_{0100} a_{2010}-a_{1000} a_{1110}$.

\vspace{-0.5cm}

{\small
\begin{flalign*}
H_{17} &=
a_{1000} x +a_{0100} y +a_{0010} z +a_{0001} w+a_{1010} x z +a_{0110} y z &&\\
&-\frac{\left(a_{0100} a_{1010}-a_{0110} a_{1000}\right) z^{2}}{a_{0001}}.
\end{flalign*}
}

\vspace{-1cm}
{\small
\begin{flalign*}H_{18}&=
a_{1000} x +a_{0100} y +a_{0200} y^{2}+a_{0010} z +a_{0001} w+a_{0110} y z +a_{0210} y^{2} z&&\\
& +\frac{2 a_{0210} a_{1000} y z^{2}}{a_{0001}}+\frac{a_{1000} \left(a_{0001} a_{0110}-a_{0200} a_{1000}\right) z^{2}}{a_{0001}^{2}}+\frac{a_{0210} a_{1000}^{2} z^{3}}{a_{0001}^{2}}.
\end{flalign*}
}

\vspace{-1cm}

{\small
\begin{flalign*}H_{19}&=
a_{1000} x +a_{1110} x y z +a_{0001} w +a_{0010} z +a_{0100} y +a_{0200} y^{2}+a_{1010} x z &&\\
&+a_{1100} x y+\frac{a_{1110} A_1^{2} z^{3}}{4 a_{0001}^{2} a_{0200} a_{1100}}+\frac{2 a_{0200} a_{1010} y z}{a_{1100}}+\frac{a_{1100}^{2} x^{2}}{4 a_{0200}}+\frac{a_{1110} a_{1100} x^{2} z}{4 a_{0200}}&&\\
&-\frac{A_1\, A_2 z^{2}}{4 a_{0001}^{2} a_{0200} a_{1100}}+\frac{a_{0200} a_{1110} y^{2} z}{a_{1100}}-\frac{a_{1110} A_1 x z^{2}}{2 a_{0001} a_{0200}}-\frac{a_{1110} A_1 y z^{2}}{a_{0001} a_{1100}}.
\end{flalign*}
}where
{ \begin{align*}
A_1&=a_{0100} a_{1100}-2 a_{0200} a_{1000},\\
A_2&=4 a_{0001} a_{0200} a_{1010}+a_{0100} a_{1100}^{2}-2 a_{0200} a_{1000} a_{1100}.
\end{align*}}

\vspace{-1cm}

{\small
\begin{flalign*}
H_{20}&=
a_{1000} x +a_{2000} x^{2}+a_{3000} x^{3}+a_{0010} z +a_{1010} x z +\frac{a_{0011}^{2} z^{2}}{4 a_{0002}}+\frac{2 a_{0010} a_{0002} w}{a_{0011}}&&\\
&+\frac{2 a_{1010} a_{0002} x w}{a_{0011}}+\frac{2 a_{0002} a_{2010} x^{2} w}{a_{0011}}+a_{0011} z w +a_{0002} w^{2}+a_{2010} x^{2} z.
\end{flalign*}
}

\vspace{-1cm}
{\small
\begin{flalign*} H_{21}&=
a_{1000} x +a_{0001} w +a_{0010} z +a_{0100} y +a_{1010} x z +\frac{2 a_{1010} a_{0002} x w}{a_{0011}}+a_{0002} w^{2}&&\\
&-\frac{\left(a_{0001} a_{0011}-2 a_{0002} a_{0010}\right)  A x^{2}}{4 a_{0002} a_{0011} a_{0100}^{2}}+\frac{a_{0011}^{2} z^{2}}{4 a_{0002}}+a_{0011} z w,
\end{flalign*}
} where $A=\left(a_{0001} a_{0011}^{2}-2 a_{0002} a_{0010} a_{0011}+4 a_{0002} a_{0100} a_{1010}\right)$.

\vspace{-0.2cm}
{\small
\begin{flalign*} H_{22} &=
a_{1000} x +a_{2000} x^{2}+\frac{a_{0011}^{2} \left(2 a_{0011} a_{0,0,3,0}+a_{0110} a_{1020}\right) x^{3}}{a_{0110}^{3}}-a_{0011} x y&&\\
& +a_{0010} z  -\frac{a_{0011} \left(3 a_{0011} a_{0,0,3,0}+2 a_{0110} a_{1020}\right) x^{2} z}{a_{0110}^{2}}+a_{0110} y z &&\\
&-\frac{a_{0110} \left(a_{0011} a_{1010}+a_{0110} a_{2000}\right) z^{2}}{a_{0011}^{2}}+a_{1020} x \,z^{2}+a_{0,0,3,0} z^{3}-\frac{a_{0011}^{2} x w}{a_{0110}}\\
&+a_{0011} z w+a_{1010} x z.
\end{flalign*}
}

\vspace{-1cm}
{\small
\begin{flalign*}H_{23}&=
a_{1000} x +a_{2000} x^{2} -a_{0011} x y -\frac{a_{0110} \left(a_{0011} a_{1000}+a_{0100} a_{2000}\right) z}{a_{0011}^{2}}&&\\
&+a_{1010} x z  -\frac{a_{0110} \left(a_{0011} a_{1010}+a_{0110} a_{2000}\right)z^{2}}{a_{0011}^{2}}-\frac{a_{0011}^{2} x w}{a_{0110}}+a_{0011} z w&&\\
&+a_{0100} y +a_{0110} y z.
\end{flalign*}
}

\vspace{-1cm}
{\small
\begin{flalign*}H_{24}&=
a_{1000} x +a_{0010} z +a_{0100} y +a_{0200} y^{2}+\frac{a_{0110} \left(a_{0011}+a_{1100}\right) x z}{2 a_{0200}}&&\\
&-\frac{\left(a_{0100} a_{1100}-2 a_{0200} a_{1000}\right)^{2} a_{0110} w x}{\left(2 a_{0010} a_{0200}-a_{0100} a_{0110}\right)^{2}}+a_{1100} x y+a_{0110} y z &&\\
&-\frac{a_{0200} \left(a_{0100} a_{1100}-2 a_{0200} a_{1000}\right)^{2} w^{2}}{\left(2 a_{0010} a_{0200}-a_{0100} a_{0110}\right)^{2}}+\frac{\left(a_{0010} a_{1100}-a_{0110} a_{1000}\right) A_1 x^{2}}{\left(2 a_{0010} a_{0200}-a_{0100} a_{0110}\right)^{2}}\\
&-\frac{A_2 \,A_3 z^{2}}{4 a_{0200} \left(a_{0100} a_{1100}-2 a_{0200} a_{1000}\right)^{2}}+a_{0011} z w,
\end{flalign*}
} where
{ \begin{align*}
A_1&=a_{0010} a_{0200} a_{1100}-a_{0100} a_{0110} a_{1100}+a_{0110} a_{0200} a_{1000},\\
A_2&=2 a_{0010} a_{0011} a_{0200}-a_{0011} a_{0100} a_{0110}-a_{0100} a_{0110} a_{1100}\\
&+2 a_{0110} a_{0200} a_{1000},\\
A_3&=2 a_{0010} a_{0011} a_{0200}-a_{0011} a_{0100} a_{0110}+a_{0100} a_{0110} a_{1100}\\
&-2 a_{0110} a_{0200} a_{1000}.
\end{align*}}

\vspace{-1cm}
{\small
\begin{flalign*}H_{25} &=
\frac{a_{0100} a_{1100} x}{2 a_{0200}}+\frac{\left(a_{0002} a_{0110}^{2}+a_{0200} a_{1100}^{2}\right) x^{2}}{4 a_{0200}^{2}}+a_{0100} y +a_{1100} x y &&\\
&+a_{0200} y^{2}+\frac{a_{0100} a_{0110} z}{2 a_{0200}}+\frac{a_{0110} \left(a_{0011}+a_{1100}\right) x z}{2 a_{0200}}+a_{0110} y z &&\\
&+\frac{\left(a_{0002} a_{0110}^{2}+a_{0011}^{2} a_{0200}\right) z^{2}}{4 a_{0002} a_{0200}}+\frac{a_{0002} a_{0110} x w}{a_{0200}}+a_{0011} z w +a_{0002} w^{2}.
\end{flalign*}

\vspace{-0.6cm}
{\small
\begin{flalign*}H_{26,27} &=
a_{1000} x +a_{0001} w +\frac{\left(a_{0001} a_{0011} a_{0200}+a_{0002} a_{0100} a_{0110}\pm\sqrt{A}\right) z}{2 a_{0002} a_{0200}} &&\\
&+a_{0100} y+a_{0200} y^{2}+\frac{a_{0110} \left(a_{0011}+a_{1100}\right) x z}{2 a_{0200}}+\frac{a_{0002} a_{0110} x w}{a_{0200}}+a_{1100} x y &&\\
&+a_{0110} y z +a_{0002} w^{2}+\frac{\left(a_{0002} a_{0110}^{2}+a_{0200} a_{1100}^{2}\right) x^{2}}{4 a_{0200}^{2}}&&\\
&+\frac{\left(a_{0002} a_{0110}^{2}+a_{0011}^{2} a_{0200}\right) z^{2}}{4 a_{0002} a_{0200}}+a_{0011} z w,
\end{flalign*}
}

where $A=-a_{0002} a_{0200} \left(a_{0001} a_{0110}+a_{0100} a_{1100}-2 a_{0200} a_{1000}\right)^{2}$.

{\small
\begin{flalign*}H_{28} &=
a_{1000} x +a_{2000} x^{2}+a_{3000} x^{3}+a_{0100} y +\frac{a_{0100} a_{1001} x y}{a_{0001}}+a_{0010} z +a_{1010} x z  &&\\
&+a_{2010} x^{2} z-\frac{a_{0100}^{2} a_{1001} y z}{a_{0001}^{2}}-\frac{\left(a_{0001} a_{1010}+a_{0100} a_{2000}\right) a_{0100} z^{2}}{a_{0001}^{2}}&&\\
&-\frac{a_{0100} \left(2 a_{0001} a_{2010}+3 a_{0100} a_{3000}\right) x z^{2}}{a_{0001}^{2}}
+a_{0001} w +a_{1001} x w &&\\
&
+\frac{a_{0100}^{2} \left(a_{0001} a_{2010}+2 a_{0100} a_{3000}\right) z^{3}}{a_{0001}^{3}}-\frac{a_{0100} a_{1001} z w}{a_{0001}}.
\end{flalign*}

\vspace{-0.5cm}

{\small
\begin{flalign*}H_{29} &=
a_{1000} x +a_{2000} x^{2}+a_{0100} y -a_{0011} x y -\frac{A z}{a_{1001}^{2}}+a_{1010} x z -\frac{a_{0011}^{2} y z}{a_{1001}}&&\\
&-\frac{a_{0011} \left(a_{0011} a_{2000}-a_{1001} a_{1010}\right) z^{2}}{a_{1001}^{2}}+a_{0001} w +a_{1001} x w +a_{0011} z w,
\end{flalign*}
}
where \begin{align*}
A=&a_{0001} a_{0011} a_{2000}-a_{0001} a_{1001} a_{1010}-a_{0011} a_{1000} a_{1001}\\
&-a_{0100} a_{1001} a_{2000}.
\end{align*}

\vspace{-0.5cm}

{\small
\begin{flalign*}H_{30} &=
\frac{\left(a_{0001} a_{0011}+a_{0100} a_{1001}\right) x}{a_{0101}}+\frac{a_{1001} \left(a_{0011} a_{0101}+a_{0200} a_{1001}\right) x^{2}}{a_{0101}^{2}}&&\\
&+a_{0100} y +\frac{\left(a_{0011} a_{0101}+2 a_{0200} a_{1001}\right) x y}{a_{0101}}+a_{0200} y^{2}+a_{0010} z &&\\
&+\frac{\left(a_{0011}^{2}+a_{0110} a_{1001}\right) x z}{a_{0101}}+\frac{a_{0210} a_{1001}^{2} x^{2} z}{a_{0101}^{2}}+a_{0110} y z +\frac{2 a_{0210} a_{1001} x y z}{a_{0101}}&&\\
&+a_{0210} y^{2} z -\frac{a_{0011} \left(a_{0011} a_{0200}-a_{0101} a_{0110}\right) z^{2}}{a_{0101}^{2}}+\frac{2 a_{0210} a_{0011} a_{1001} x z^{2}}{a_{0101}^{2}}&&\\
&+\frac{2 a_{0210} a_{0011} y z^{2}}{a_{0101}}+\frac{a_{0011}^{2} a_{0210} z^{3}}{a_{0101}^{2}}+a_{0001} w +a_{1001} x w +a_{0101} y w +a_{0011} z w.
\end{flalign*}
}

\vspace{-0.8cm}
{\small
\begin{flalign*}H_{31} &=
a_{1000} x +\frac{a_{1001} A_1 x^{2}}{a_{0101}^{2}}-\frac{A_2 y}{a_{1001}}+\frac{A_3 x y}{a_{0101}}+a_{0200} y^{2}-\frac{A_4 z}{a_{0101}^{2} a_{1001}}&&\\
&+\frac{\left(a_{0011}^{2}+a_{0110} a_{1001}\right) x z}{a_{0101}}+\frac{a_{0210} a_{1001}^{2} x^{2} z}{a_{0101}^{2}}+a_{0110} y z +\frac{2 a_{0210} a_{1001} x y z}{a_{0101}}&&\\
&+a_{0210} y^{2} z +\frac{2 a_{0210} a_{0011} a_{1001} x z^{2}}{a_{0101}^{2}}-\frac{a_{0011} \left(a_{0011} a_{0200}-a_{0101} a_{0110}\right) z^{2}}{a_{0101}^{2}}&&\\
&+\frac{2 a_{0210} a_{0011} y z^{2}}{a_{0101}}+\frac{a_{0011}^{2} a_{0210} z^{3}}{a_{0101}^{2}}+a_{0001} w +a_{1001} x w +a_{0101} y w +a_{0011} z w,
\end{flalign*}
} where \begin{align*}
A_1&=a_{0011} a_{0101}+a_{0200} a_{1001},\\
A_2&=a_{0001} a_{0011}-a_{0101} a_{1000},\\
A_3&=a_{0011} a_{0101}+2 a_{0200} a_{1001},\\
A_4&=a_{0001}^{2} a_{0210} a_{1001}+a_{0001} a_{0011}^{2} a_{0101}+2 a_{0001} a_{0011} a_{0200} a_{1001}\\
&-a_{0001} a_{0101} a_{0110} a_{1001}-a_{0011} a_{0101}^{2} a_{1000}.
\end{align*}

\vspace{-0.6cm}
{\small
\begin{flalign*}H_{32} &=
a_{1000} x -\frac{a_{1100} A x^{2}}{a_{0101}^{2}}+a_{0100} y +a_{1100} x y +\frac{a_{0100} a_{1100} z}{a_{0101}}+\frac{a_{1100}^{2} x z}{a_{0101}}&&\\
&+a_{0001} w +a_{1001} x w +a_{0101} y w +a_{1100} z w +a_{0002} w^{2},
\end{flalign*}
}where $A=a_{0002} a_{1100}-a_{0101} a_{1001}$.

\vspace{-0.4cm}
{\small
\begin{flalign*}H_{33} &=
\frac{A x}{a_{0101}^{2}}-\frac{a_{1100} \left(a_{0002} a_{1100}-a_{0101} a_{1001}\right) x^{2}}{a_{0101}^{2}}+a_{0100} y +a_{1100} x y +a_{0010} z &&\\
&+\frac{\left(a_{0110} a_{1001}+a_{1100}^{2}\right) x z}{a_{0101}}+a_{0110} y z +\frac{a_{0110} \left(a_{0002} a_{0110}+a_{0101} a_{1100}\right) z^{2}}{a_{0101}^{2}}&&\\
&+a_{0001} w +a_{1001} x w +a_{0101} y w +\frac{\left(2 a_{0002} a_{0110}+a_{0101} a_{1100}\right) z w}{a_{0101}}+a_{0002} w^{2},
\end{flalign*}
}where \begin{align*}A&=a_{0001} a_{0002} a_{0110}+a_{0001} a_{0101} a_{1100}-a_{0002} a_{0010} a_{0101}-a_{0002} a_{0100} a_{1100}\\
&+a_{0100} a_{0101} a_{1001}.\end{align*}

\vspace{-0.6cm}
{\small
\begin{flalign*}H_{34} &=
a_{1000} x -\frac{a_{1100} \left(a_{0002} a_{1100}-a_{0101} a_{1001}\right) x^{2}}{a_{0101}^{2}}+a_{0100} y +a_{1100} x y +a_{0010} z &&\\
&+\frac{\left(a_{0110} a_{1001}+a_{1100}^{2}\right) x z}{a_{0101}}+a_{0110} y z +\frac{a_{0110} \left(a_{0002} a_{0110}+a_{0101} a_{1100}\right) z^{2}}{a_{0101}^{2}}&&\\
&+\frac{\left(a_{0010} a_{0101}-a_{0100} a_{1100}\right) w}{a_{0110}}+a_{1001} x w +a_{0101} y w+a_{0002} w^{2} &&\\
&+\frac{\left(2 a_{0002} a_{0110}+a_{0101} a_{1100}\right) z w}{a_{0101}}.
\end{flalign*}
}

\vspace{-0.6cm}
{\small
\begin{flalign*}H_{35} &=
a_{1000} x +\frac{a_{1001} \left(a_{0011} a_{0101}+a_{0200} a_{1001}\right) x^{2}}{a_{0101}^{2}}+a_{0100} y+a_{0200} y^{2}-\frac{A z}{a_{0101}^{2}}&&\\
& +\frac{\left(a_{0011} a_{0101}+2 a_{0200} a_{1001}\right) x y}{a_{0101}}+\frac{\left(a_{0011}^{2}+a_{0110} a_{1001}\right) x z}{a_{0101}}+a_{0110} y z &&\\
&-\frac{a_{0011} \left(a_{0011} a_{0200}-a_{0101} a_{0110}\right) z^{2}}{a_{0101}^{2}}+a_{0001} w +a_{1001} x w +a_{0101} y w&&\\& +a_{0011} z w,
\end{flalign*}
}where \begin{align*}A&=a_{0001} a_{0011} a_{0200}-a_{0001} a_{0101} a_{0110}-a_{0011} a_{0100} a_{0101}-a_{0100} a_{0200} a_{1001}\\
&+a_{0101} a_{0200} a_{1000}.\end{align*}

\vspace{-0.6cm}
{\small
\begin{flalign*}H_{36} &=
a_{1000} x +a_{0100} y +a_{0200} y^{2}+a_{0010} z +a_{0001} w +a_{0101} y w&&\\
& -\frac{a_{1000} \left(a_{0010} a_{0101}+a_{0200} a_{1000}\right) w^{2}}{a_{0010}^{2}}.
\end{flalign*}
}

\vspace{-1cm}
{\small
\begin{flalign*}H_{37} &=
a_{0002} w^{2}+a_{0101} y w +a_{0200} y^{2}+a_{0001} w +a_{0100} y.&&\\\end{flalign*}
}

\vspace{-1.4cm}
{\small
\begin{flalign*}H_{38} &=
a_{1000} x +a_{2000} x^{2}-\frac{a_{0101} a_{2010} x^{3}}{6 a_{0200}}+\frac{2 a_{0001} a_{0200} y}{a_{0101}}+a_{0200} y^{2}&&\\
&-\frac{2 a_{0200} a_{1000} z}{a_{0101}}-\frac{4 a_{0200} a_{2000} x z}{a_{0101}}+a_{2010} x^{2} z +\frac{4 a_{0200}^{2} a_{2000} z^{2}}{a_{0101}^{2}}&&\\
&-\frac{2 a_{0200} a_{2010} x z^{2}}{a_{0101}}+\frac{4 a_{0200}^{2} a_{2010} z^{3}}{3 a_{0101}^{2}}+a_{0001} w +a_{0101} y w +\frac{a_{0101}^{2} w^{2}}{4 a_{0200}}.\end{flalign*}
}

\vspace{-0.8cm}
{\small
\begin{flalign*}H_{39} &=
a_{1000} x +\frac{a_{0101}^{2} a_{1020} x^{3}}{12 a_{0200}^{2}}+\frac{2 a_{0001} a_{0200} y}{a_{0101}}+a_{0200} y^{2}-\frac{2 a_{0200} a_{1000} z}{a_{0101}}&&\\
&-\frac{a_{0101} a_{1020} x^{2} z}{2 a_{0200}}+a_{1020} x z^{2}-\frac{2 a_{0200} a_{1020} z^{3}}{3 a_{0101}}+a_{0001} w +a_{0101} y w &&\\
&+\frac{a_{0101}^{2} w^{2}}{4 a_{0200}}.\end{flalign*}
}

\vspace{-1cm}
{\small
\begin{flalign*} H_{40} &=
a_{1000} x +a_{0010} z +a_{0100} y +a_{0101} y w +\frac{a_{0101}^{2} y^{2}}{4 a_{0002}}+\frac{a_{0101} A^{2} z x}{4 a_{0002}^{2} a_{0100}^{2}}+a_{0002} w^{2}&&\\
&-\frac{A^{2} x^{2}}{4 a_{0002} a_{0100}^{2}}-\frac{a_{0101}^{2} A^{2} z^{2}}{16 a_{0002}^{3} a_{0100}^{2}}.\end{flalign*}
}where $A=2 a_{0002} a_{0010}+a_{0101} a_{1000}$.

\vspace{-0.4cm}
{\small
\begin{flalign*} H_{41} &=
a_{1000} x +a_{0001} w +a_{0010} z +a_{0100} y +\frac{2 a_{0200} a_{1001} y w}{a_{1100}}+a_{0200} y^{2}+a_{1001} x w &&\\
&+a_{1100} x y +\frac{\left(a_{0100} a_{1100}-2 a_{0200} a_{1000}\right) A w^{2}}{4 a_{0010}^{2} a_{0200} a_{1100}}&&\\
&+\frac{a_{1100}^{2} x^{2}}{4 a_{0200}},\end{flalign*}
}where
$A=4 a_{0010} a_{0200} a_{1001}-a_{0100} a_{1100}^{2}+2 a_{0200} a_{1000} a_{1100}. $

\vspace{-0.4cm}
{\small
\begin{flalign*}
H_{42} &=
a_{1000} x +\frac{a_{1100}^{2} x^{2}}{4 a_{0200}}+\frac{2 a_{0200} a_{1000} y}{a_{1100}}+a_{1100} x y +a_{0200} y^{2}+a_{0001} w &&\\
&+a_{1001} x w +\frac{2 a_{0200} a_{1001} y w}{a_{1100}}+a_{0002} w^{2}.\end{flalign*}
}

\vspace{-1cm}
{\small
\begin{flalign*} H_{43} &=
a_{1000} x +a_{0001} w +a_{0010} z +a_{0100} y +a_{0101} y w +a_{0200} y^{2}+\frac{a_{0101} a_{1100} x w}{2 a_{0200}}&&\\
&+a_{1100} x y+\frac{a_{0101} \left(a_{0010} a_{0101}-a_{0100} a_{1100}+2 a_{0200} a_{1000}\right)^{2} z x}{\left(2 a_{0001} a_{0200}-a_{0100} a_{0101}\right)^{2}} +\frac{a_{0101}^{2} w^{2}}{4 a_{0200}}&&\\
&+\frac{A_1 \, A_2 x^{2}}{4 a_{0200} \left(2 a_{0001} a_{0200}-a_{0100} a_{0101}\right)^{2}}-\frac{a_{0200} A_3^{2} z^{2}}{\left(2 a_{0001} a_{0200}-a_{0100} a_{0101}\right)^{2}},\end{flalign*}
}where
\begin{align*}
A_1&=2 a_{0001} a_{0200} a_{1100}-a_{0010} a_{0101}^{2}-2 a_{0101} a_{0200} a_{1000},\\
A_2&=2 a_{0001} a_{0200} a_{1100}+a_{0010} a_{0101}^{2}-2 a_{0100} a_{0101} a_{1100}+2 a_{0101} a_{0200} a_{1000},\\
A_3&=a_{0010} a_{0101}-a_{0100} a_{1100}+2 a_{0200} a_{1000}.
\end{align*}

\vspace{-0.6cm}
{\small
\begin{flalign*} H_{44} &=
a_{1000} x -\frac{\left(a_{0101} a_{1010}-a_{1100}^{2}\right) x^{2}}{4 a_{0200}}+a_{0100} y +a_{1100} x y +a_{0200} y^{2}+a_{1010} x z &&\\
&-\frac{a_{0200} a_{1010} z^{2}}{a_{0101}}+\frac{\left(a_{0100} a_{1100}-2 a_{0200} a_{1000}\right) z}{a_{0101}}+\frac{a_{0100} a_{0101} w}{2 a_{0200}}&&\\
&+\frac{a_{0101} a_{1100} x w}{2 a_{0200}}+a_{0101} y w +\frac{a_{0101}^{2} w^{2}}{4 a_{0200}}.\end{flalign*}
}

\vspace{-0.6cm}
{\small
\begin{flalign*} H_{45} &=
\frac{a_{0100} a_{1100} x}{2 a_{0200}}+\frac{\left(2 a_{0200} a_{1010} a_{2010}+a_{1020} a_{1100}^{2}\right) x^{2}}{4 a_{0200} a_{1020}}+\frac{a_{2010}^{2} x^{3}}{3 a_{1020}}+a_{0100} y &&\\
&+a_{1100} x y +a_{0200} y^{2}+a_{1010} x z +a_{2010} x^{2} z +\frac{a_{1010} a_{1020} z^{2}}{2 a_{2010}}+a_{1020} x z^{2}&&\\
&+\frac{a_{1020}^{2} z^{3}}{3 a_{2010}}-\frac{a_{0100} a_{2010} w}{a_{1020}}-\frac{a_{2010} a_{1100} x w}{a_{1020}}-\frac{2 a_{0200} a_{2010} y w}{a_{1020}}&&\\
&+\frac{a_{0200} a_{2010}^{2} w^{2}}{a_{1020}^{2}}.\end{flalign*}
}

\vspace{-0.8cm}
{\small
\begin{flalign*} H_{46} &=
a_{1000} x +\frac{\left(2 a_{0200} a_{1010} a_{2010}+a_{1020} a_{1100}^{2}\right) x^{2}}{4 a_{0200} a_{1020}}+\frac{a_{2010}^{2} x^{3}}{3 a_{1020}}
+a_{1100} x y &&\\
&+a_{0200} y^{2}+a_{0010} z -\frac{2 a_{0200} \left(a_{0010} a_{2010}-a_{1000} a_{1020}\right) y}{a_{1020} a_{1100}}
+a_{1010} x z &&\\
&+a_{2010} x^{2} z +\frac{a_{1010} a_{1020} z^{2}}{2 a_{2010}}+a_{1020} x z^{2}+\frac{a_{1020}^{2} z^{3}}{3 a_{2010}}-\frac{a_{2010} a_{1100} x w}{a_{1020}}&&\\
&+\frac{2 a_{0200} \left(a_{0010} a_{2010}-a_{1000} a_{1020}\right) a_{2010} w}{a_{1020}^{2} a_{1100}}-\frac{2 a_{0200} a_{2010} y w}{a_{1020}}&&\\
&+\frac{a_{0200} a_{2010}^{2} w^{2}}{a_{1020}^{2}}.\end{flalign*}
}

\vspace{-0.6cm}
{\small
\begin{flalign*}H_{47} &=
a_{1000} x +\frac{a_{1100}^{2} x^{2}}{4 a_{0200}}+\frac{a_{0101}^{2} a_{1020} x^{3}}{12 a_{0200}^{2}}+a_{0100} y +a_{1100} x y +a_{0200} y^{2}&&\\
&-\frac{a_{0101} a_{1020} x^{2} z}{2 a_{0200}}+\frac{\left(a_{0100} a_{1100}-2 a_{0200} a_{1000}\right) z}{a_{0101}}+a_{1020} x z^{2}&&\\
&-\frac{2 a_{0200} a_{1020} z^{3}}{3 a_{0101}}+\frac{a_{0100} a_{0101} w}{2 a_{0200}}+\frac{a_{0101} a_{1100} x w}{2 a_{0200}}+a_{0101} y w +\frac{a_{0101}^{2} w^{2}}{4 a_{0200}}.\end{flalign*}
}

\vspace{-0.6cm}
{\small
\begin{flalign*}H_{48} &=
a_{1000} x +a_{0001} w +a_{0010} z +a_{0101} y w +\frac{a_{0101}^{2} y^{2}}{4 a_{0002}}+\frac{A_1^{2} z x}{a_{0101}^{3} a_{0001}^{2}}+a_{0110} y z &&\\
&+a_{0002} w^{2}-\frac{a_{0002} A_1^{2} x^{2}}{a_{0001}^{2} a_{0101}^{4}}+\frac{ A_2\, A_3 z^{2}}{4 a_{0101} a_{0001}^{2} a_{0002}}+\frac{2 a_{0002} a_{0110} z w}{a_{0101}},\end{flalign*}
}where
\begin{align*}
A_1&=2 a_{0001} a_{0002} a_{0110}-2 a_{0002} a_{0010} a_{0101}-a_{0101}^{2} a_{1000},\\
A_2&=2 a_{0002} a_{0010}+a_{0101} a_{1000},\\
A_3&=4 a_{0001} a_{0002} a_{0110}-2 a_{0002} a_{0010} a_{0101}-a_{0101}^{2} a_{1000}.
\end{align*}

\vspace{-0.7cm}
{\small
\begin{flalign*}H_{49} &=
a_{1000} x +a_{2000} x^{2}-\frac{a_{1010} a_{0101} y^{2}}{4 a_{2000}}+\frac{a_{1000} a_{1010} z}{2 a_{2000}}+a_{1010} x z +a_{0110} y z &&\\
&+\frac{\left(a_{0101} a_{1010}^{3}-4 a_{0110}^{2} a_{2000}^{2}\right) z^{2}}{4 a_{0101} a_{1010} a_{2000}}+a_{0101} y w -\frac{2 a_{2000} a_{0110} z w}{a_{1010}}&&\\
&-\frac{a_{0101} a_{2000} w^{2}}{a_{1010}}.\end{flalign*}
}

\vspace{-0.8cm}
{\small
\begin{flalign*}H_{50} &=
a_{1000} x +a_{2000} x^{2}+\frac{2 a_{2000} a_{2010} x^{3}}{3 a_{1010}}+\frac{a_{0002} a_{1010}^{2} y^{2}}{4 a_{2000}^{2}}+\frac{a_{1000} a_{1010} z}{2 a_{2000}}&&\\
&+a_{1010} x z +a_{2010} x^{2} z +a_{0110} y z +\frac{\left(a_{0002} a_{1010}^{4}+4 a_{0110}^{2} a_{2000}^{3}\right) z^{2}}{4 a_{0002} a_{2000} a_{1010}^{2}}&&\\
&+\frac{a_{1010} a_{2010} x \,z^{2}}{2 a_{2000}}+\frac{a_{1010}^{2} a_{2010} z^{3}}{12 a_{2000}^{2}}-\frac{a_{0002} a_{1010} y w}{a_{2000}}-\frac{2 a_{2000} a_{0110} z w}{a_{1010}}&&\\
&+a_{0002} w^{2}.\end{flalign*}
}

\vspace{-1.2cm}
{\small
\begin{flalign*}H_{51} &=
a_{0200} y^{2}+a_{0110} y z +\frac{a_{0002} a_{0110}^{2} z^{2}}{a_{0101}^{2}}+a_{0101} y w +\frac{2 a_{0002} a_{0110} z w}{a_{0101}}+a_{0002} w^{2}.&&\\\end{flalign*}
}

\vspace{-1.2cm}
{\small
\begin{flalign*}H_{52} &=
a_{1000} x +a_{0200} y^{2}+a_{0010} z +a_{0110} y z -\frac{a_{1000} A a_{0110}^{2} z^{2}}{a_{0010}^{2} a_{0101}^{2}}+a_{0101} y w &&\\
&-\frac{2 a_{1000} A a_{0110} z w}{a_{0010}^{2} a_{0101}}-\frac{a_{1000} A w^{2}}{a_{0010}^{2}},\end{flalign*}
} where $A=a_{0010} a_{0101}+a_{0200} a_{1000}$.

\vspace{-0.2cm}
{\small
\begin{flalign*}H_{53} &=
a_{1000} x -\frac{a_{0101} a_{2010} x^{3}}{6 a_{0200}}+a_{0200} y^{2}-\frac{2 a_{0200} a_{1000} z}{a_{0101}}+a_{2010} x^{2} z +a_{0110} y z &&\\
&+\frac{a_{0110}^{2} z^{2}}{4 a_{0200}}-\frac{2 a_{0200} a_{2010} x \,z^{2}}{a_{0101}}+\frac{4 a_{0200}^{2} a_{2010} z^{3}}{3 a_{0101}^{2}}+a_{0101} y w +\frac{a_{0101} a_{0110} z w}{2 a_{0200}}&&\\
&+\frac{a_{0101}^{2} w^{2}}{4 a_{0200}}.\end{flalign*}
}

\vspace{-0.6cm}
{\small
\begin{flalign*}H_{54} &=
a_{1000} x +a_{0001} w +a_{0101} y w +a_{0200} y^{2}+a_{1010} x z +\frac{a_{1010} A_1 x w}{a_{0110} A_5^{2}}&&\\
&+a_{0110} y z +\frac{A_2\, A_3 w^{2}}{a_{0001}^{2} a_{0110}^{2} A_5^{2}}-\frac{a_{0200} a_{0001}^{2} a_{1010}^{2} x^{2}}{A_5^{2}}+\frac{2 A_4 z w}{a_{0001}^{2} a_{0110}}&&\\
&+\frac{a_{1000} \left(a_{0001} a_{0110}-a_{0200} a_{1000}\right) z^{2}}{a_{0001}^{2}},\end{flalign*}
} where
\begin{align*}
A_1&=a_{0001}^{2} a_{0101} a_{0110}^{2}-4 a_{0001}^{2} a_{0200}^{2} a_{1010}-4 a_{0001} a_{0101} a_{0110} a_{0200} a_{1000}\\
&+4 a_{0101} a_{0200}^{2} a_{1000}^{2},\\
A_2&=a_{0001}^{2} a_{0101} a_{0110}^{2}-2 a_{0001}^{2} a_{0200}^{2} a_{1010}-3 a_{0001} a_{0101} a_{0110} a_{0200} a_{1000}\\
&+2 a_{0101} a_{0200}^{2} a_{1000}^{2},\\
A_3&=2 a_{0001}^{2} a_{0200} a_{1010}+a_{0001} a_{0101} a_{0110} a_{1000}-2 a_{0101} a_{0200} a_{1000}^{2},\\
A_4&=a_{0001}^{2} a_{0200} a_{1010}+a_{0001} a_{0101} a_{0110} a_{1000}-a_{0101} a_{0200} a_{1000}^{2}\\
A_5&=a_{0001} a_{0110}-2 a_{0200} a_{1000}.\\\end{align*}

\newpage

\section{Appendix}\label{s:appendixc}

We present here a $10$-parameter family of third-order Hamiltonian systems, for which their associated KHK map preserves the original Hamiltonian. Within this family there are cases where the associated vector field is not a Lie symmetry of the KHK map. 

\begin{proposition}\label{p:propoappendixc} Consider the Hamiltonian function
\begin{flalign}
H =&  b_{1} x^{3}+3  b_{1} x^{2} y + b_{2} x^{2} z 
+3  b_{1} x  y^{2}+2  b_{2} x y z 
+\left(b_{5} b_{6} -b_{3} b_{4}\right) x  z^{2} 
+ b_{6} x z s + b_{4} x z r + b_{1} y^{3}\nonumber&&\\
&+ b_{2} y^{2} z +\left(b_{5} b_{6}-b_{3} b_{4}\right) y  z^{2}+ b_{6} y z s + b_{4} y z r 
+ b_{7} x^{2}+2  b_{7} x y + b_{8} x z
 + x w +b_{3}  x s \label{e:Hpropo9}&&\\
& +b_{5}  x r
 + b_{7} y^{2}+ b_{8} y z + y w
 +b_{3}  y s +b_{5}  y r +b_{9} b_{5} z^{2}
+b_{9}  z s +b_{10}  x  +b_{10}  y.\nonumber
\end{flalign}
\noindent The corresponding Hamiltonian vector field $X$ of the form \eqref{e:ham6} is not a Lie Symmetry of its  associated
KHK maps if and only if $b_{4}b_{9}\neq 0$.
\end{proposition}

\begin{proof}
The KHK map associated to the vector field $X$, whose expression we omit but which is straightforward to obtain using \eqref{e:ham6}, is $\Phi_h=(\Phi _{1,h},\Phi _{2,h},\Phi _{3,h},\Phi _{4,h},\Phi _{5,h},\Phi _{6,h})$ with

\vspace{-0.5cm}

\begin{flalign*}
\Phi_{1,h}=&\left( h(2 h  b_{2}-12   b_{1}) x^{2}+h(4 h  b_{2}-24   b_{1}) x y -h(h^{2} b_{9}  b_{4}
+8   b_{2}) x z
+h(2 h  b_{2}-12   b_{1}) y^{2} \right.&&\\
&-h(h^{2} b_{9}  b_{4}+8   b_{2}) y z
 +h(2 h b_{9}  b_{4}+4  b_{3} b_{4}-4  b_{5} b_{6}) z^{2}
-4 h  b_{4} z r -4 h  b_{6} z s &&\\
&+(-8 h  b_{7}+4 ) x-8 h  b_{7} y -h(2 h b_{9} b_{5} +4   b_{8}) z -4 h  w -4 h b_{5}  r -h(2 h b_{9} +4  b_{3} ) s&&\\
&\left. -4 h b_{10} \right)/4,
\end{flalign*}

\vspace{-1cm}

\begin{flalign*}
\Phi_{2,h}=&\left(h(-2 h  b_{2}+12   b_{1}) x^{2}+h(-4 h  b_{2}+24   b_{1}) x y +h(h^{2} b_{9}  b_{4}+8   b_{2}) x z +h(-2 h  b_{2}+12   b_{1}) y^{2}\right.&&\\
&+h(h^{2} b_{9}  b_{4}+8   b_{2}) y z +h(-2 h b_{9}  b_{4}-4  b_{3} b_{4}+4  b_{5} b_{6}) z^{2}+4 h  b_{4} z r +4 h  b_{6} z s +8 h  b_{7} x&&\\
&\left.  +(8 h  b_{7}+4 ) y +h(2 h b_{9} b_{5} +4  b_{8}) z +4 h  w +4 h b_{5}  r +h(2 h b_{9} +4  b_{3} ) s +4 h b_{10} \right)/4,
\end{flalign*}

\vspace{-1cm}

\begin{flalign*}
\Phi_{3,h}&=-h  x -h  y +z,&&\\
\end{flalign*}

\vspace{-1.5cm}

\begin{flalign*}
\Phi_{4,h}=&\left(h( [b_{3}  b_{4}- b_{5}  b_{6}]h+2   b_{2}) x^{2}+h([2  b_{3}  b_{4}-2  b_{5}  b_{6}]h+4  b_{2}) x y  +h(4  b_{5} b_{6}-4  b_{3} b_{4}) x z \right. &&\\
&+2 h  b_{4} x r +2 h  b_{6} x s +h([b_{3}  b_{4}- b_{5}  b_{6}]h+2   b_{2}) y^{2}+h(4  b_{5} b_{6}-4  b_{3} b_{4}) y z +2 h  b_{4} y r &&\\
&\left.+2 h  b_{6} y s +h(-h b_{9} b_{5} +2   b_{8}) x +h(-h b_{9} b_{5} +2   b_{8}) y +4 h b_{9} b_{5} z +2  w +2 h b_{9}  s\right)/2,
\end{flalign*}

\vspace{-1cm}

\begin{flalign*}
\Phi_{5,h}=&\left(h^{2}  b_{6} x^{2}+2 h^{2}  b_{6} x y -2 h b_{6} x z +h^{2}  b_{6} y^{2}-2 h b_{6} y z + h(h b_{9} -2  b_{3}) x  + h(h b_{9} -2  b_{3}) y \right.&&\\
&\left.-2 h b_{9} z +2 r\right)/2,
\end{flalign*}

\vspace{-1cm}

\begin{flalign*}
\Phi_{6,h}=&\left(-h^{2}  b_{4} x^{2}-2 h^{2}  b_{4} x y +2 h b_{4} x z-h^{2}  b_{4} y^{2}+2 h b_{4} y z +2 h b_{5} x  +2 h b_{5} y +2 s\right)/2.&&\\
\end{flalign*}

\vspace{-0.5cm}

A straightforward computation shows that 
$$X{|\Phi_h}-{\mathrm D}\Phi_h\,X=
\left( \frac{b_{4} b_{9}  \left(x +y \right)^{2}h^{3} }{4}, -\frac{b_{4} b_{9}  \left(x +y \right)^{2}h^{3} }{4} , 0 , 0 , 0 , 0 \right)^t,
$$
and the result follows.
\end{proof}

\end{document}